\documentclass[reqno,11pt,centertags]{amsart}
\usepackage{amsmath,amsthm,amscd,amssymb,latexsym}
\usepackage{upref,enumerate,hyperref}
\date{\today}

\numberwithin{equation}{section}

\input epsf
\usepackage{epsfig}
\usepackage[T2A,OT1]{fontenc}
\usepackage[ot2enc]{inputenc}
\usepackage[russian,english]{babel}
\usepackage{epic,eepicemu}
\usepackage[all]{xy}
\usepackage[center]{caption}

\usepackage{xcolor}

\newcommand{\bbD}{{\mathbb{D}}}
\newcommand{\bbN}{{\mathbb{N}}}
\newcommand{\bbR}{{\mathbb{R}}}

\newcommand{\bbZ}{{\mathbb{Z}}}
\newcommand{\bbC}{{\mathbb{C}}}

\newcommand{\cA}{{\mathcal{A}}}
\newcommand{\cB}{{\mathcal{B}}}
\newcommand{\cC}{{\mathcal{C}}}
\newcommand{\cD}{{\mathcal{D}}}
\newcommand{\cE}{{\mathcal{E}}}

\newcommand{\cH}{{\mathcal{H}}}

\newcommand{\cK}{{\mathcal{K}}}
\newcommand{\cL}{{\mathcal{L}}}
\newcommand{\cM}{{\mathcal{M}}}
\newcommand{\cN}{{\mathcal{N}}}

\newcommand{\cQ}{{\mathcal{Q}}}
\newcommand{\cR}{{\mathcal{R}}}

\newcommand{\cT}{{\mathcal{T}}}

\newcommand{\cX}{{\mathcal{X}}}

\newcommand{\bA}{{\mathbf{A}}}
\newcommand{\bB}{{\mathbf{B}}}

\newcommand{\fA}{{\mathfrak{A}}}
\newcommand{\fB}{{\mathfrak{B}}}

\newcommand{\fj}{{\mathfrak{j}}}

\newcommand{\ba}{{\mathbf{a}}}
\newcommand{\bb}{{\mathbf{b}}}

\newcommand{\E}{{\mathsf{E}}}
\newcommand{\sE}{{\mathsf{E}}}
\newcommand{\sF}{{\mathsf{F}}}

\newcommand{\PSU}{{\mathrm{PSU}}}

\newcommand{\G}{{\Gamma}}

\renewcommand{\a}{\alpha}
\renewcommand{\b}{\beta}
\newcommand{\g}{\gamma}
\newcommand{\e}{{\epsilon}}

\newcommand{\s}{{\sigma}}
\newcommand{\vk}{\varkappa}
\newcommand{\vt}{\vartheta}
\renewcommand{\l}{\lambda}
\renewcommand{\L}{\Lambda}
\newcommand{\om}{{\omega}}
\newcommand{\z}{\zeta}
\def\u{\upsilon}

\newcommand{\vp}{\varphi}

\newcommand{\pd}{{\partial}}

\renewcommand{\Re}{\text{\rm Re}\,}
\renewcommand{\Im}{\text{\rm Im}\,}

\newcommand{\tr}{\text{\rm tr}\,}

\newcommand{\ess}{\text{\rm ess}}
\newcommand{\Res}{\text{\rm Res}\,}

\newtheorem{theorem}{Theorem}[section]

\newtheorem{lemma}[theorem]{Lemma}
\newtheorem{proposition}[theorem]{Proposition}

\newtheorem{corollary}[theorem]{Corollary}

\newtheorem{conjecture}[theorem]{Conjecture}
\theoremstyle{definition}
\newtheorem{definition}[theorem]{Definition}

\newtheorem{remark}[theorem]{Remark}
\newtheorem{problem}[theorem]{Problem}

\date{\today}

\title[The Deift Conjecture]{ The Deift Conjecture: \\ A Program to Construct a Counterexample}

\author[D.\ Damanik]{David Damanik}
\address{Department of Mathematics, Rice University, Houston, TX 77005, USA}
\email{damanik@rice.edu}

\author[M.\ Luki\'c]{Milivoje Luki\'c}
\address{Department of Mathematics, Rice University, Houston, TX 77005, USA}
\email{lukic@rice.edu}

\author[A.\ Volberg]{Alexander Volberg}
\address{Department of Mathematics, Michigan State University, East Lansing, MI 48824, USA}
\email{volberg@msu.edu}

\author{Peter Yuditskii}

\address{Institut f\"ur Analysis, Johannes Kepler Universit\"at Linz, 4020 Linz, Austria}

\email{peter.yuditskii@gmail.com}

\begin{document}


\maketitle

\tableofcontents

\newpage

\section{Introduction}

Deift proposed the following conjecture:

\begin{conjecture}[Deift \cite{De,De2}]
For the KdV equation $\partial_t u - 6 u \partial_x u + \partial_x^3 u = 0$ with almost periodic initial data $u(x,0) = V(x)$,  the solution evolves almost periodically in time.
\end{conjecture}

The conjecture is motivated by the case of periodic initial data, for which it was proved by McKean--Trubowitz \cite{McKT}. In recent years, it has been proved for certain classes of reflectionless initial data \cite{BDGL1,EVY}.

However, in this manuscript  we present a program which intends to show that the conjecture is not true in general, by constructing  $C^\infty$ almost periodic initial data $V$ for which the solution is not almost periodic in $t$.\footnote{This program was presented in a lecture by P.~Y.\ at Johannes Kepler Universit\"at on November 16, 2021 (https://www.jku.at/en/institute-of-analysis/conferences/on-the-deift-question/), and the present manuscript expands on the material presented in that lecture.}

The key element of our construction is a \textit{dichotomy} that was found in \cite{VYInv}, see also \cite{DY}. In a moment we will recall such notions as
\begin{itemize}
\item[1.1] Widom domain
\item[1.2] DCT property
\item[1.3] Reflectionless 1-D Schr\"odinger operators
\end{itemize}
Using these terms, the VY-dichotomy is the following statement.
\medskip

\noindent
\textit{For a generic spectral set $\sE$ such that the domain $\Omega=\bbC\setminus \sE$ is of Widom type, the following dichotomy holds:
\begin{itemize}
\item[a)] If the DCT property holds in $\Omega$, then all reflectionless potentials $V$ with spectrum $\sE$ are almost periodic.
\item[b)] If DCT fails, no reflectionless potential with spectrum $\sE$ is almost periodic.
\end{itemize}
}

\textbf{Our main attempt} is to construct a \textit{special} spectral set $\sE$ such that DCT fails but
there exists an almost periodic reflectionless potential $V(x)$ with this spectral set; on the other hand the translation in time should be \textit{generic}, for an exact setting see Problem~\ref{prob13sept21}. Then the VY-dichotomy implies that the solution of the KdV equation with this almost periodic initial datum is not almost periodic in time.

\subsection{Widom Domain}

For our purpose it is enough to consider comparably simple spectral sets of the form
\begin{equation}\label{setE1}
\E = [0,\infty) \setminus \cup_{j=1}^\infty (a_j, b_j),
\end{equation}
where
\begin{equation}\label{setE2}
0 \leftarrow \dots < a_n < b_n < \dots <  a_2 < b_2 < a_1 < b_1.
\end{equation}
Thus, the spectral set $\sE$ represents a system of proper intervals $[b_{k+1},a_k]$, $k\ge 1$,
 accumulating to the origin only, together with the point of accumulation as well as the half axis $[b_1,\infty)$. We emphasize that the condition  \eqref{setE2} on the set $\sE$ in \eqref{setE1} holds throughout the whole manuscript.

We assume that $\Omega = \bbC \setminus \E$ is a Dirichlet regular domain. For any $z_0 \in \Omega$, we denote by $G(z,z_0) = G_\Omega(z,z_0)$ the \emph{Green function} with a logarithmic pole at $z_0$ for the domain $\Omega$; by Dirichlet regularity, $G(\cdot,z_0)$ is continuous on $\Omega$ and vanishes continuously on $\E \cup \{\infty\}$. In our case only the origin has to be checked in this sense for regularity.

\begin{definition}
Let $\tilde c_j$ be the critical points of $G(z,-1)$, i.e., $\nabla G(\tilde c_j,-1)=0$.  We say that $\Omega$ is of \emph{Widom type} if
\begin{equation}\label{11sept21}
\sum_{j\ge 1}G(\tilde c_j,-1)<\infty.
\end{equation}
\end{definition}

Note  that in our case each gap contains a unique critical point $\tilde c_j\in(a_j,b_j)$, $j\ge 1$.

The key property of Widom domains is the following.
Consider multi-valued meromorphic functions $f$ on $\Omega$ such that $\lvert f \rvert$ is single-valued. For such functions, there exists a character $\alpha=\alpha_f : \pi_1(\Omega) \to \bbR / \bbZ$ such that
\[
f \circ \gamma = e^{2\pi i \alpha(\gamma)} f, \qquad \forall \gamma \in \pi_1(\Omega).
\]
The function $f$ is said to be \emph{character-automorphic}.
By $H^\infty(\a)$ we understand  the collection of bounded character-automorphic functions with the given character $\a$.

Due to Widom, \eqref{11sept21} holds if and only if $H^\infty(\a)$ contains a non-constant function for all characters $\a$.

\subsection{DCT Property}

Statements about character-automorphic functions can also be viewed in terms of lifts to the universal cover $\bbD$ via the uniformization $\Omega \simeq \bbD / \Gamma$. More precisely, we denote by $\Lambda : \bbD / \Gamma \to \Omega$ a uniformization of $\Omega$, where $\Gamma \cong \pi_1(\Omega)$ is a Fuchsian group (discrete subgroup of $\PSU(1,1)$). In particular, $\L$ is surjective and $\Lambda(\z_1) = \L(\z_2)$ if and only if $\z_2 = \g(\z_1)$ for some $\g \in \G$. Now the multi-valued function $f$ lifts to a single-valued function $F$ on $\bbD$; formally, $F = f \circ \L$.

The function $f$ is said to have \emph{bounded characteristic} if its lift $F$ has bounded characteristic, i.e., if $F = F_1 / F_2$ for some $F_1, F_2  \in H^\infty(\bbD)$. If $F_2$ is outer, $f$ is said to be of \emph{Smirnov class}. This class of character-automorphic functions is denoted by $\cN_+(\Omega)$.

\begin{definition}
Let $f$ be single-valued in the domain such that $f\in N_+(\Omega)$ and
\begin{equation}\label{13sept21}
\int_{\sE}(|f(\l+i0)|+|f(\l-i0|)\frac{d\l}{\l+1}<\infty.
\end{equation}
We say that DCT (Direct Cauchy Theorem) holds in $\Omega$ if for all such $f$'s, we have
$$
f(z)=\frac 1 {2\pi i}\oint_{\pd\Omega }\frac{f(\l)d\l}{\l-z}.
$$

\end{definition}

\subsection{Reflectionless Operators}\label{subs1_3}

For any $V \in L^\infty(\bbR)$, the Schr\"odinger operator $L_V = -\partial_x^2 +V$ is a self-adjoint operator on $L^2(\bbR)$ with domain $D(L_V) = W^{2,2}(\bbR)$. In particular, it is limit point at $\pm \infty$, and for any $z \notin \sigma_\ess(L_V)$, there are Weyl solutions $\psi_\pm(x,z)$ which solve the eigenfunction equation
\[
- \partial_x^2 \psi_\pm (x,z) + V(x) \psi_\pm(x,z) = z \psi_\pm(x,z),
\]
are nontrivial, and  square integrable at $\pm \infty$, respectively. These Weyl solutions are defined uniquely up to normalization, and their logarithmic derivatives
\[
m_\pm(x,z) = \pm \frac{\partial_x \psi_\pm(x,z)}{\psi_\pm(x,z)},
\]
viewed as functions of $z$, are Herglotz functions (analytic maps of $\bbC_+$ to itself). Moreover, $m_\pm(x,\cdot)$ are spectral functions corresponding to half-line restrictions of $L_V$ to half-lines $[x,\pm \infty)$ with a Dirichlet boundary condition at $x$, and by the Borg--Marchenko theorem, they determine $V$ uniquely.

\begin{definition}
We will write $m_\pm (z) = m_\pm(0,z)$.
The potential $V$ is said to be \emph{reflectionless} if
\begin{equation}\label{reflectionless}
\lim_{\epsilon\downarrow 0} m_+(\lambda + i\epsilon) = - \lim_{\epsilon\downarrow 0} \overline{ m_-(\lambda + i\epsilon) }
\end{equation}
for Lebesgue-a.e.\ $\lambda \in \sigma(L_V)$. We denote by $\cR(\E)$ the set of all reflectionless potentials with $\sigma(L_V) = \E$.
\end{definition}

\section{Introduction: Details of the Plan}\label{sect2}

\subsection{Divisors (Dirichlet Data) and the (Generalized) Abel Map}

Following \cite{SY97}, we associate with a given $V\in\cR(E)$ a divisor $D\in\cD(\sE)$ to which we can apply the Abel map acting from $\cD(\sE)$ to the group of characters $\pi_1(\Omega)^*$.
This Abel map linearizes each flow in the KdV hierarchy.

\begin{definition}\label{def127}
To each gap $(a_j,b_j)$, $j\ge 1$, we associate a topological circle
\[
I_j = [a_j, b_j] \times \{ -1, +1 \} /_{\substack{(a_j,-1)\sim (a_j, +1) \\ (b_j,-1)\sim (b_j, +1)}}
\]
The product $\cD(\E) =  \prod_{j=1}^\infty I_j$ is called the \emph{torus of divisors} associated to $\sE$ or the \emph{torus of Dirichlet data}.
\end{definition}

Thus $D\in\cD(\sE)$ is the collection
$$
D = \{ (\lambda_j, \epsilon_j) \}_{j=1}^\infty,
$$
where $\l_j\in[a_j,b_j]$ and $\e_j=\pm 1$ with the mentioned identification
$$
(a_j,-1)\sim (a_j, +1),\quad (b_j,-1)\sim (b_j, +1).
$$

\begin{lemma} There is a homeomorphism
$$
\mathcal{B} : \cR(\E) \to \cD(\E), \; V \mapsto D.
$$
\end{lemma}

Characters of the group $\pi_1(\Omega)$ form a group, which we denote by $\pi_1(\Omega)^*$. The group $\pi_1(\Omega)$ has a collection of free generators: loops $\g_k$ starting at $-1$ and going through the gap $(a_k,b_k)$. Thus each character $\a$ can be identified with a sequence
$\{\a_k\}_{k\ge 1}$,
$\a_k\in\bbR/\bbZ$, where
$$
\a_k=\a(\g_k).
$$
In this way we fix an identification $\pi_1(\Omega)^*\simeq(\bbR/\bbZ)^\infty$.

\begin{definition}\label{def227}
Let $\omega(\sF,z)$ be the harmonic measure of the set $\sF\subset \sE$ in the domain w.r.t. $z\in\Omega$  and $\sE_k=\sE\cap[0,a_k]$. The \emph{Abel map} $\cA:D(\sE)\to\pi_1(\Omega)^*$ is given by
\begin{equation}\label{12sept21}
\cA_k(D)=\frac 1 2 \sum_{j\ge 1}\left(\omega(\l_j,\sE_k)-\omega(b_j,\sE_k)\right)\e_j \mod 1.
\end{equation}
\end{definition}

\begin{conjecture}\label{th13sept21}
The composition map $\cA\circ\cB:\cR(\sE)\to\pi_1(\Omega)^*$ linearizes the shifts of a reflectionless potential both in space and time, that is, there exists $\eta=\{\eta_k\}_{k\ge 1}$ such that
$$
\cA(\cB(V(x+x_0))=\{\a_k-\eta_k x_0\mod 1\}_{k\ge 1}
$$
and there exists  $\eta^{(1)}=\{\eta^{(1)}_k\}_{k\ge 1}$ such that
$$
\cA(\cB(u(x,t)))=\{\a_k-\eta^{(1)}_k t \mod 1\}_{k\ge 1}
$$
where $u(x,t)$ is the solution of the KdV equation with $u(x,0)=V(x)$.
\end{conjecture}

We refer the reader to \cite{EVY}, but emphasize that our setting is outside the formal scope of the main results of \cite{EVY} since our domains do not satisfy DCT.
In Sect.~\ref{sect5} we develop a technique that in principle should allow one to remove this restriction. Based on this technique we prove Proposition~\ref{prop13sept21} supporting the conjecture in the part dealing with the space shift.

\subsection{Minimal Violations of DCT and the Main Lemma}

The following lemma is an easy consequence of the Morera theorem \cite{Gar}.

\begin{lemma}\label{l116jun21}
Let the single-valued $f\in N_+(\Omega)$ obey \eqref{13sept21}.
Then
$$
h_f(z)=\frac 1 {2\pi i}\oint_{\pd\Omega }\frac{f(\l)d\l}{\l-z}-f(z)
$$
is an entire function w.r.t. $1/z$.
\end{lemma}

We say that $\Omega$ is a domain with a \textit{minimal violation of DCT} if $h_f$ is a polynomial of $1/z$ for an arbitrary admissible function $f$.
Recall that DCT holds if $h_f=0$ for an arbitrary such function $f$.

For Widom domains, the Abel map is well defined and it is continuous. DCT is a criterion for uniqueness \cite{SY97, VYInv}. Thus, as soon as DCT is violated, we have a set $\Xi$ of \textit{singular} characters, which we define as follows,
\begin{equation}\label{23sept21}
\Xi=\{\a\in\pi_1(\Omega)^*: \ \a=\cA(D_1)=\cA(D_2),\quad D_1\not=D_2\}.
\end{equation}

Finding a description of the singular set $\Xi$ we consider as one of the most interesting and delicate questions in the modern development of Widom theory. In particular, it underlies our proposed solution of the Deift problem.

\begin{theorem}\label{l13sept21}
Let $\Omega$ be a Widom domain with a minimal violation of DCT. Assume that
\begin{equation}\label{33sept21}
\int_{\sE\cap[0,b_1]}\frac{d\l}{\l}<\infty
\end{equation}
but
\begin{equation}\label{43sept21}
\int_{\sE}\frac{d\l}{\l^2}=\infty.
\end{equation}
Then $\a\in\Xi$ if and only if there exists $D$ such that $\a=\cA(D)$ and
 \begin{equation}\label{53sept21}
\sum\log\frac{b_j}{\l_j}<\infty.
\end{equation}
\end{theorem}

In what follows,
\begin{equation}\label{129sept21}
\cD_b(\sE)=\left\{D\in \cD(\sE):\ \sum\log\frac{b_j}{\l_j}<\infty\right\}.
\end{equation}

Note that the divisor $D_0=\{b_j\}_{j\ge 1}$ is mapped to the origin of $\pi_{1}(\Omega)^*$, see \eqref{12sept21}. Thus $\a\in\Xi$ if the character is close to the origin in a certain (precise) sense,
\eqref{53sept21}, and this is definitely an analytic condition on the character.

We say that a direction $\eta=\{\eta_k\}_{k\ge 1}$ is \emph{generic} if its coordinates are rationally independent. Otherwise
\begin{equation}\label{63sept21}
\cT_\eta=\text{clos}_{x\in\bbR}\{\a(x)=\{\a_k(x)\}:\ \a_k(x)=\eta_k x\mod 1\}
\end{equation}
is a proper subtorus of $(\bbR/\bbZ)^\infty$, or in our identification of $\pi_1(\Omega)^*$.

By Conjecture \ref{th13sept21} (in fact we will use Proposition \ref{prop13sept21}), the translation flow maps via $\mathcal{A} \circ \mathcal{B}$ to the translation flow in the direction $-\eta$, that is, with
$$
S_x : \cR(\E) \to \cR(\E), \; V(\cdot) \mapsto V(\cdot + x),
$$
we have $(\mathcal{A} \circ \mathcal{B})(S_x V) = \alpha_V - \eta x$, where $\alpha_V := (\mathcal{A} \circ \mathcal{B})(V)$ is the character corresponding to $V$.
Let us consider the orbit closure
$$
 \text{clos}\{ \alpha_V - \eta x : x \in \mathbb{R} \}=\{\b=\a_V+\a:\ \a\in\cT_\eta\}=\cT_\eta(\a_V).
$$

\begin{lemma}[Main Lemma]\label{mainl}
Suppose that $V \in \cR(\E)$ is such that
$$
\cT_\eta(\a_V) \cap \Xi = \emptyset, \quad \a_V=(\cA\circ\cB)(V).
$$
Then $V$ is almost periodic.
\end{lemma}

The proof of this lemma is simple, but ideologically this is the main point in our construction. It allows us to reduce the Deift problem to the following one.

\begin{problem}\label{prob13sept21}
Construct a Widom domain $\Omega=\bbC\setminus\sE$ such that
\begin{itemize}
\item[i)] DCT fails, i.e., $\Xi\not=\emptyset$,
\item[ii)] $\eta^{(1)}$ is in a generic position,
\item[iii)] there exists $\b\in\pi_1(\Omega)^*$ such that
$$
\cT_\eta(\b)\cap\Xi=\emptyset,\quad \cT_\eta(\b):=\{\b+\a:\ \a\in\cT_\eta\},
$$
\end{itemize}
where $\eta$ and $\eta^{(1)}$ are defined in Conjecture \ref{th13sept21}
(an explicit definition is given in the next section in terms of the so-called comb domains).
\end{problem}

\begin{remark}
iii) is naturally the main requirement here. ii) should be a consequence of a certain stability of iii) with respect  to small perturbations of the spectral set, see Conjecture~\ref{con16sept}.  If DCT holds, then $\Xi=\emptyset$ and iii) trivially holds for all $\b$. But if $\Xi\not=\emptyset$, then $\eta$ has to be degenerate, since in the generic case $\cT_\eta=\pi_1(\Omega)^*$ and the intersection cannot be empty.
\end{remark}

Most likely $\Xi$ is always dense in $\pi_1(\Omega)^*$. At least this is definitely true under the assumptions of Theorem~\ref{l13sept21}, since by \eqref{53sept21} any finite number of elements of $D$ can be modified in an arbitrary way.

Finally, in this section we would like to provide a simplest possible \textit{model example}.
The exact statements and conjectures are formulated in the next section.

Consider the set
\begin{equation}\label{16oct21}
\Xi^0=\{\a: \lim_{k\to \infty}\a_k=0\}.
\end{equation}
Note that the set is given in terms of a certain analytic condition and it is evidently dense in $\pi_1(\Omega)^*$ as well as $\Xi$. Also, if we consider the set of divisors
$$
\cD^0_b(\sE)=\{D:\ \exists N_D\  \text{such that }\ \l_j=b_j\ \text{for all}\ j\ge N_D\},
$$
then $\cA(\cD^0_b(\sE))\subset \Xi^0$.

\begin{proposition}\label{pr16sept21}
Let $\eta_k=1/3^{k-1}$, $k\ge 1$. Then $\cT_\eta\simeq \bbZ_3\times(\bbR/\bbZ)$.
There exists $\b\in\pi_1(\Omega)^*$ such that $\cT_{\eta}(\b)\cap\Xi^0=\emptyset$.
\end{proposition}

Our expectation with respect to the real shape of $\cT_\eta$ is stated in Lemma~\ref{l16sept21}. For this $\cT_\eta$, we still have $\cT_{\eta}(\b)\cap\Xi^0=\emptyset$ for some $\b\in\pi_1(\Omega)^*$ as an easy consequence of Proposition \ref{pr16sept21}. However, we do not expect that $\Xi\subset\Xi^0$.

\section{Introduction: Results and Conjectures}

\subsection{Comb Domains. The Stability Conjecture}

Conformal mappings on the so-called comb domains (combs) are a convenient tool in the inverse spectral theory of periodic and almost periodic operators.

With a system of parameters $\eta=\{\eta_k\}_{k\ge 1}$,
\begin{equation}\label{16sept21}
0 \leftarrow \dots <\eta_{k+1}<\eta_k<\dots<\eta_1<\infty
\end{equation}
and $\{h_k\}_{k\ge 1}$,
\begin{equation}\label{26sept21}
\sum h_k<\infty
\end{equation}
we associate the domain
\begin{equation}\label{36sept21}
\Pi=\bbC_{++}\setminus\cup_{k\ge 1}\{w=\eta_k+ih,\quad h\le h_k\},
\end{equation}
where $\bbC_{++}$ is the quarter plane
$$
\bbC_{++}=\{w:\ \Im w>0,\ \Re w>0\}.
$$

Let $\Theta:\ \bbC_+\to\Pi$ be the conformal mapping normalized by
$$
\Theta(z)\sim i\sqrt\l,\quad z=-\l,\ \l\to\infty,\quad \Theta(0)=0.
$$
Let $\sE=\Theta^{-1}(\bbR_+)$. Then $\Omega=\bbC\setminus\sE$ is a Widom domain such that $\sE$ obeys \eqref{setE2}. The conditions \eqref{16sept21} and \eqref{26sept21} provide a parametric description of such domains. The function $\Im\Theta(z)$ has a single valued extension in $\Omega$, this is the so called \emph{Martin function} of this domain with respect to $\infty$.

For this $\sE$, the parameters $\{\eta_k\}$ are exactly the ones that were introduced in Conjecture~\ref{th13sept21}. The parameters $\eta^{(1)}=\{\eta^{(1)}_k\}$ deal with the following comb,
\begin{equation}\label{46sept21}
\Pi^{(1)}=\bbC^3_{++}\setminus\cup_{k\ge 1}\{w=\eta^{(1)}_k+ih,\quad h\le h^{(1)}_k\},
\end{equation}
a subdomain of the 3/4-plane $\bbC^3_{++}=\{w^3:\ w\in\bbC_{++}\}$. Together with the accompanied $\{h_k^{(1)}\}$ under the conditions
\begin{equation}\label{56sept21}
0 \leftarrow \dots <\eta^{(1)}_{k+1}<\eta^{(1)}_k<\dots<\eta^{(1)}_1<\infty
\end{equation}
and
\begin{equation}\label{66sept21}
\sum h^{(1)}_k < \infty,
\end{equation}
they provide a parametric description of the Widom domain $\Omega=\bbC\setminus\sE$ obeying \eqref{setE2} via the conformal mapping
 $\Theta^{(1)}:\ \bbC_+\to\Pi^{(1)}$  normalized by
$$
\Theta^{(1)}(z)\sim -i\sqrt\l^3,\quad z=-\l,\ \l\to\infty,\quad \Theta^{(1)}(0)=0.
$$
with $\sE=\left(\Theta^{(1)}\right)^{-1}(\bbR_+)$.

\begin{conjecture}\label{con16sept}
With a fixed collection $\eta$, one can bring $\eta^{(1)}$ in a generic position with an arbitrary small variation of $\{h_k\}_{k\ge 1}$.
\end{conjecture}

In Sect.~\ref{sect4} we show that in the case of a finite number of gaps, an arbitrary variation of the frequencies $\{\eta,\eta^{(1)}\}$ can be uniquely achieved by a suitable variation of the gap end-points $\{a_k,b_k\}_{k\ge 1}$, see Theorem \ref{thmain4}.

\subsection{The Simplest Violation of DCT. $\cT_\eta$ as a Profinite Completion of $\bbZ$.}\label{sbs32}

The origin, as the unique accumulation point of the spectral intervals, plays a special role in our construction. We define one more comb related to the Martin function in the domain with respect to the origin.

It is convenient to set $\l=1/\mu^2$. We define $\Delta(\mu)$ as a conformal mapping of the upper half plane $\bbC_+$ onto the symmetric comb
$$
\Pi_0=\bbC_+\setminus\cup_{k\ge 1}\{w=\pm n_k+i \u:\ \u\le\u_k\}.
$$
The condition
$$
0 <n_1<\dots<n_k<n_{k+1}<\dots\to\infty
$$
deals with \eqref{setE2}, however the sequence $\{\u_k\}_{k\ge1}$ is possibly even unbounded in the case of interest to us, contrary to \eqref{26sept21}. We denote by $(\mu_k^-,\mu_k^+)$ the preimage of the slit $\{w= n_k+i\u:\ \u\le\u_k\}$. Respectively,
$$
(-\mu_k^+,-\mu_k^-)=\Delta^{-1}(-n_k+i\u)\quad\text{for}\ \u\le\u_k,
$$
and $a_k=1/(\mu_k^+)^2$, $b_k=1/(\mu_k^-)^2$.

The function
\begin{equation}\label{76sept21}
\Theta_0(z)=\Delta(\mu), \quad z=1/\mu^2, \ \mu\in\bbC_{-+}=\{\Re\mu<0,\ \Im\mu>0\},
\end{equation}
maps the upper half plane onto the left part of the comb
$$
\Pi^-_0=\{w\in\Pi_0:\ \Re w<0\}.
$$
In this case $\Im\Theta_0(z)$ is the Martin function w.r.t. the origin in $\Omega=\bbC\setminus\sE$ with
$\sE=\bbR_+\setminus\cup_{k\ge 1}(a_k,b_k)$.

\begin{proposition}\label{pr17sept21}
Let
\begin{equation}\label{130jun21}
\cos\pi \Delta_\rho(\mu)=\cos\pi\mu-\rho\mu \sin\pi \mu,\quad \rho>0.
\end{equation}
$\Delta_\rho(\mu)$ provides a conformal mapping of $\bbC_+$ on the domain $\Pi_0$, whose frequencies are given explicitly by
$n_k=k$, $k\not=0$.
Further, let
$$
\cE=\cE_\rho=\left\{z=\frac 1{\mu^2}:\ |\cos\pi\Delta_\rho(\mu)|\le 1\right\}.
$$
Then $\Omega_\cE=\bbC\setminus \cE$ is of Widom type with a minimal DCT violation and the set $\cE$
obeys
 \eqref{33sept21},
\eqref{43sept21}.
\end{proposition}

Having this model spectral set $\cE$, we conjecture the following (for a more detailed motivation, see Sect.~\ref{sect7}).

\begin{conjecture}\label{conj33} There exists $\Omega$ of Widom type with a minimal DCT violation and $\sE$ obeying \eqref{33sept21},
\eqref{43sept21} such that $\eta_k=1/k$.
\end{conjecture}

For the given system of space-frequencies generators, we have the following description of the torus $\cT_\eta$.

\begin{lemma}\label{l16sept21}
Let $\eta_k=1/k$. Then $\cT_\eta\simeq(\bbR/\bbZ)\times\hat \bbZ$,
where $\hat \bbZ$ is the profinite completion of integers
$$
\hat \bbZ=\prod_{p\text{\ prime}}\bbZ_p.
$$
\end{lemma}

\subsection{Description of Singular Characters}

By Theorem \ref{l13sept21} we have a description of the singular set $\Xi$ in terms of divisors \eqref{53sept21}. Using Lemma \ref{l16sept21}, to check iii) in Proposition \ref{prop13sept21} we need to describe $\Xi$ directly in terms of characters. The absence of such a description of $\Xi$ seems  to be the main obstacle in disproving the Deift conjecture.

The spectral set $\cE$ defined in Proposition \ref{pr17sept21} is remarkable in the sense that its parameters are quite explicitly given both in the spectral plane and in terms of potential theory (comb parameters).
\begin{problem}
For the set $\cE$ defined via $\Delta(\mu)$ in \eqref{130jun21}, translate the characteristic property \eqref{53sept21} for $\a\in\Xi$ into the language of potential theory (characteristics in the comb domains).
\end{problem}

Our working hypothesis is that a description of $\Xi$ in this case can be given in terms of the convergence of the series
\begin{equation}\label{17sept21}
\sum_{k\ge 1}\frac{\a_k}{k}.
\end{equation}
However at present we are even unable to clarify what kind of convergence should be claimed here (absolute or conditional; possibly $\a\in\Xi$ just requires the boundedness of partial sums in \eqref{17sept21}).

\bigskip
\noindent
\textbf{Acknowledgments.}
D.~D.\ was supported in part by Simons Fellowship $\# 669836$ and NSF grants DMS--1700131 and DMS--2054752. M.~L.\ was supported in part by NSF grant DMS--1700179. A.~V.\ was supported in part by NSF grant DMS--1900286. P.~Y.\ was supported by the Austrian Science Fund FWF, project no: P32885-N.

\section{Stability Conjecture \ref{con16sept}. Finite Number of Gaps}\label{sect4}

In the classical finite dimensional case, our considerations deal with the Riemann surface
$$
\cR=\{(z,w):\ w^2=T(z)\},\quad T(z)=z\prod_{k=1}^n(z-a_k)(z-b_k),
$$
compactified by  a single point at infinity. We use the standard dissection of our hyperelliptic Riemann surface $\cR$: the $\bA_k$-cut starts at $-1$ on the upper sheet of the surface and goes back though the gap $(a_k, b_k)$, see Fig. 4.1. The $\bB_k$-cut goes along the gap $(a_k, b_k)$ on the upper sheet and goes back to the initial point along the same interval on the lower sheet.

\begin{center}
\includegraphics[scale=0.8]{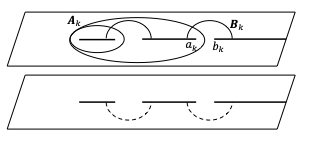}

Fig. 4.1.  Dissection of $\cR_*:=\cR\setminus\cup_{j=1}^n(\bA_j\cup\bB_j)$
\end{center}

Recall the following classical relation, see e.g. Theorem 4 \cite[  Ch. 9, \S 2]{CH}, for abelian integrals.
Let $\theta_{1,2}$  be abelian integrals on $\cR$ with the periods
$$
(\ba_k)_{1,2}=\int_{\bA_k}d \theta_{1,2},\quad (\bb_k)_{1,2}=\int_{\bB_k}d \theta_{1,2}.
$$
We assume that the poles of the differentials $d\theta_{1,2}$ are out of the curves  $\bA_k$ and $\bB_k$. Let $\cR_*:=\cR\setminus\cup_{j=1}^n(\bA_j\cup\bB_j)$.
Then
\begin{equation}\label{114sept21}
\frac 1{2\pi i}\int_{\pd\cR_*}\theta_1d\theta_2=
\frac 1{2\pi i}\sum_{k=1}^n((\ba_k)_1(\bb_k)_2-(\bb_k)_1(\ba_k)_2)=\sum\Res \theta_1d\theta_2.
\end{equation}

In our considerations, as special role is played by the functions
$\Theta(z)$ and $\Theta^{(1)}(z)$, given by the abelian integrals
\begin{equation}\label{615sept21}
\Theta(z)=\frac 1 2\int_0^z\frac{P(\l)\, d\l}{\sqrt{T(\l)}},\quad \Theta^{(1)}(z)=\frac 3 2\int_0^z\frac{ Q(\l)\, d\l}{\sqrt{T(\l)}},
\end{equation}
where
$ Q(z)=zP^{(1)}(z)$, and $P$ and $P^{(1)}$ are monic polynomials of degree $n$. They are defined uniquely by the integral conditions
\begin{equation}\label{29jul21}
\int_{\bB_k}d\Theta(z)=2\int_{a_k}^{b_k} d\Theta(z)=0,\quad
\int_{\bB_k}d\Theta^{(1)}(z)=
2\int_{a_k}^{b_k} d\Theta^{(1)}(z)=0.
\end{equation}
In turn,
\begin{equation}\label{39jul21}
\int_{\bA_k}d\Theta(z)=2\int_{0}^{a_k} d\Theta(z)=2\eta_k,\quad
\int_{\bA_k}d\Theta^{(1)}(z)=
2\int_{0}^{a_k} d\Theta^{(1)}(z)=2\eta_k^{(1)}.
\end{equation}

Also we will use the system of  abelian integrals of the first kind $\vt_{\E_k}$, $k=1,\dots, n$, with the normalization
\begin{equation}\label{115sept21}
\frac 1{2\pi i}\int_{\bB_j}d\vt_{\sE_k}(z)=\delta_{k, j}, \quad d\vt_{\sE_k}(z)=\frac{O^{(k)}(z)\, dz}{\sqrt{T(z)}},
\end{equation}
where $O^{(k)}$ is a suitable polynomial of degree $n-1$.

Note that the differentials $d\Theta$, $d\Theta^{(1)}$ as well as $d\vt_{\sE_k}$ are antisymmetric with respect to the substitution $(z,w)\mapsto (z,-w)$ on the surface.

\begin{remark}
We point out that in fact, in the KdV theory
for the abelian integral  describing the time evolution,
the following normalization at infinity is commonly accepted
$$
\Theta^{(t)}
=\frac 1{\tau^3}+\text{regular},\quad z=\frac 1{\tau^2}.
$$
That is,
$$
\Theta^{(t)}=\Theta^{(1)}-C\Theta
$$
with a suitable constant $C$. Let us compute this constant.
Let
$$
T(z)=z^{2n+1}+ A_T z^{2n}+\dots \quad\text{and}\quad Q(z)=z^{n+1}+ A_Q z^{n}+\dots
$$
For the differential $d\Theta^{(1)}$, we have the following local decomposition at infinity,
$$
d\Theta^{(1)}=\frac 3 2z\frac{1+A_Q/z+\dots}{\sqrt{1+A_T/z+\dots}}\frac{dz}{\sqrt{z}}=
(3z+3(A_Q-A_T/2)+\dots) \, d\sqrt{z}.
$$
In the local coordinate we have
$$
d\Theta^{(1)}=\left(-\frac{3}{\tau^4}-3\frac{A_Q-A_T/2}{\tau^2}+\text{regular}\right)d\tau,
$$
and the abelian integral $\Theta^{(1)}$ has the following leading terms,
\begin{equation}\label{314sept21}
\Theta^{(1)}=\frac{1}{\tau^3}+
3\left(
 A_Q-\frac 1 2 A_T\right)\frac 1{\tau}+\dots .
\end{equation}
Thus, since $\Theta=1/\tau+\dots$,
$$
C=3\left(
 A_Q-\frac 1 2 A_T\right)
$$
and
\begin{equation}\label{522sept21}
\Theta^{(t)}:=\Theta^{(1)}-
3\left(A_Q-\frac 1 2 A_T\right)
\Theta=\frac 1{\tau^3}+\text{regular}.
\end{equation}
Let
$$
d\Theta^{(t)}=\frac 3 2\frac{Q^{(t)}(z)}{\sqrt{T(t)}}\, dz,
$$
where, according to \eqref{522sept21},  $Q^{(t)}(z)$ is the monic polynomial
\begin{equation}\label{123sept21}
Q^{(t)}(z)=Q(z)-3\left(A_Q-\frac 1 2 A_T\right)P(z).
\end{equation}
Respectively, the frequencies $\{\eta^{(t)}_k\}_{k= 1}^n$ describing the  time evolution are given by
\begin{equation}\label{122sept21}
\eta^{(t)}_k=\int_0^{a_k}d\Theta^{(t)}=\eta^{(1)}_k
-3\left(A_Q-\frac 1 2 A_T\right)\eta_k.
\end{equation}

\end{remark}

As a simplest application of the bilinear relation  \eqref{114sept21} we prove the following well known identities.
\begin{lemma}
Let $A_{O^{(k)}}$ be the leading coefficient of the polynomial  $O^{(k)}(z)$,
$O^{(k)}(z)=A_{O^{(k)}} z^{n-1}+\dots$. Then
 $\eta_k=-A_{O^{(k)}}$.
\end{lemma}

\begin{proof}
Since infinity is the unique point of singularity for $\vt_{\sE_k}d\Theta$, by
 \eqref{29jul21} and \eqref{115sept21} we get
$$
\frac 1{2\pi i}\int_{\pd\cR_*}\vt_{\sE_k}d\Theta=-\int_{\bA_k}d\Theta=-2\eta_k=\Res_\infty\vt_{\sE_k}d\Theta.
$$
To compute the residue at infinity we use the local coordinate $z=1/\tau^2$. We have
\begin{equation}\label{414sept21}
\vt_{\sE_k}=-2A_{O^{(k)}}\tau+\dots, \quad d\Theta=-d\tau/\tau^2+\dots
\end{equation}
Thus
$$
\Res_\infty\vt_{\sE_k}d\Theta=2A_{O^{(k)}}.
$$
and the lemma is proved.
\end{proof}

Assume that all parameters vary  depending on a single variable $\e$. The derivative w.r.t. the spectral parameter is denoted by prime $(\cdot)'$ and w.r.t.  $\e$ by dot $\dot{(\cdot)}$.
We derive a system of differential equations defining the variation of the surface $\cR(\e)$ depending on variations of the system of  frequencies
$\{\dot\eta_k\}_{k=1}^n$ and $\{\dot\eta^{(t)}_k\}_{k= 1}^n$.

\begin{proposition}\label{pr120sept21}
Let $P$, $Q^{(t)}$ and $O^{(k)}$ be the polynomials forming $\Theta$, $\Theta^{(t)}$ and abelian integrals of the first kind $\vt_{\sE_k}$, respectively, see \eqref{615sept21},
\eqref{123sept21} and \eqref{115sept21}.
Let $\{\eta_k\}_{k=1}^n$ and $\{\eta^{(t)}_k\}_{k= 1}^n$ be, respectively, the space and time generating frequencies in the KdV evolution. Then
\begin{align}\label{822sept21}
\sum_{j=1}^n\left(
\frac{O^{(k)}(a_j)P(a_j)}{{T'(a_j)}}\dot a_j
+\frac{O^{(k)}(b_j)P(b_j)}{{T'(b_j)}}\dot b_j
\right)=&-2\dot \eta_k,
\\
\label{622sept21}
3\sum_{j=1}^n\left(
\frac{O^{(k)}(a_j)Q^{(t)}(a_j)}{{T'(a_j)}}\dot a_j
+\frac{O^{(k)}(b_j)Q^{(t)}(b_j)}{{T'(b_j)}}\dot b_j
\right)=&-2\dot \eta^{(t)}_k.
\end{align}
In particular, for $\dot\eta_k=0$, the following equations hold,
\begin{align}
\sum_{j=1}^n\left(\frac{O^{(k)}(a_j)P(a_j)}{T'(a_j)}\dot a_j+
\frac{O^{(k)}(b_j)P(b_j)}{T'(b_j)}\dot b_j
\right) & = 0,\label{715sept21}
\\
3\sum_{j=1}^n\left(\frac{O^{(k)}(a_j)Q(a_j)}{T'(a_j)}\dot a_j+
\frac{O^{(k)}(b_j)Q(b_j)}{T'(b_j)}\dot b_j
\right)
& = -2\dot\eta^{(t)}_k.\label{815sept21}
\end{align}
\end{proposition}

\begin{proof}
Since
$$
d\Theta=\frac 1 2\frac{P(z)}{\sqrt{T(z)}} \, dz,
$$
we have
$$
d\dot \Theta=\frac 1 2\left\{\frac{\dot P(z)}{\sqrt{T(z)}}-\frac 1 2\frac{P(z)}{\sqrt{T(z)}}
\frac{\dot T(z)}{T(z)}
\right\}dz.
$$
Since $P(z)$ is monic, the first term is an abelian differential of the first kind. The points of singularity in the second term are the $a_j$'s and $b_j$'s. It is easy to see that infinity in this case is a regular point. Since $\dot T(0)=0$, the differential is also regular at the origin.

The periods of $d\dot \Theta$ are easy to compute,
$$
\int_{\bA_k} d\dot \Theta=\pd_\e \int_{\bA_k} d \Theta=2\dot\eta_k,\quad
\int_{\bB_k} d\dot \Theta=\pd_\e \int_{\bB_k} d \Theta=0.
$$
Therefore, by \eqref{114sept21},
\begin{equation}\label{222sept21}
\frac 1{2\pi i}\int_{\pd\cR_*}\vt_{\sE_k}d\dot \Theta=-2\dot\eta_k=
\sum_{\cup_{j\ge 1}\{a_j,b_j\}}\Res \vt_{\sE_k}d\dot \Theta.
\end{equation}

To compute the residues at $a_j$ we pass to the local coordinate $\tau^2=z-a_j$. For the leading term of $d\dot \Theta$, we have
$$
d\dot \Theta = -\frac 1 4\frac{P(a_j)}{\sqrt{T'(a_j)}\tau}\frac{\dot T (a_j)}{T'(a_j)\tau ^2} 2\tau \, d\tau+\dots=
 -\frac 1 2\frac{P(a_j)}{\sqrt{T'(a_j)}}\frac{\dot T (a_j)}{T'(a_j)}\frac{d\tau}{\tau^2}+\dots .
$$
Since $T(a_j)=0$, we have $T'(a_j)\dot a_j+\dot T(a_j)=0$. Therefore,
\begin{equation}\label{322sept21}
 d\dot \Theta=\frac 1 2\frac{P(a_j)}{\sqrt{T'(a_j)}}\dot a_j \frac{d\tau}{\tau ^2}+\dots .
\end{equation}
Note that the residue (the coefficient before $d\tau/\tau$) of this differential is trivial just for the parity reason.
Similarly,
$$
d\vt_{\sE_k}=\frac{O^{(k)}(a_j)}{\sqrt{T'(a_j)}\tau} 2\tau \, d\tau+\dots = 2\frac{O^{(k)}(a_j)}{\sqrt{T'(a_j)}} d\tau+\dots .
$$
Therefore, locally,
$$
\vt_{\sE_k}=\vt_{\sE_k}(a_j)+2\frac{O^{(k)}(a_j)}{\sqrt{T'(a_j)}}\tau+\dots .
$$
In combination with \eqref{322sept21}, we get
$$
\Res_{a_j} \vt_{\sE_k}d\dot \Theta=\frac{O^{(k)}(a_j)P(a_j)}{{T'(a_j)}}\dot a_j.
$$
Since the computation at $b_j$ is the same,  \eqref{222sept21} can be rewritten as
\eqref{822sept21}.

A featured property of $d\Theta^{(t)}$ is that $d\dot \Theta^{(t)}$ is also regular at infinity, see the second expression in
\eqref{522sept21}. Therefore a word by word repetition of the above arguments provides
\eqref{622sept21}.
With $\dot\eta_k=0$, \eqref{822sept21} and \eqref{622sept21} are the same as \eqref{715sept21} and \eqref{815sept21}.
\end{proof}

Let us prove a specific property of the roots of the polynomials $P$ and $Q$.
\begin{lemma}\label{l120sept21}
Let $P$ and $Q$ be defined by \eqref{615sept21}. Their zeros $c_j$ and $c^{(1)}_j$ form interlacing sequences. In other words, $\Im Q(z)/P(z)\ge 0$ in the upper half plane.
\end{lemma}

\begin{proof}
First we show that all preimages $(Q/P)^{-1}(\vk)$, for $\vk\in\bbR$, are real.
Consider $\Theta^{(1)}_\vk=\Theta^{(1)}-\vk\Theta$, $\vk\in\bbR$.  Since $\Theta^{(1)}_\vk(a_j)=\Theta^{(1)}_\vk(b_j)=0$, each gap contains at least one critical point, that is, a zero of $Q_\vk=Q-\vk P$. Since the number of gaps is $n$ and the degree of $Q_\vk$ is $n+1$, there is one more zero of $Q_\vk$, which was not localized yet. But since $Q_\vk$ is real on the real axis, this remaining zero is also real.

Assume now that $P/Q$ is not monotonic between two consecutive zeros of $Q$. Then this function has a critical point and at least one non-real direction (which starts at this critical point) in the complex plane where $P/Q$ is still real. This contradicts the statement proven in the above paragraph. Thus, the zeros of $Q$ and $P$ are interlacing. Recall that  $P$ and $Q$ are monic polynomials, that is,
$$
\frac{Q(z)}{P(z)}\simeq z,\quad z\to\infty.
$$
By the well known property of functions with interlacing zeros and poles,
$\Re (Q/P)(z)\ge 0$ in the upper half plane and
$$
\frac{P(c_j^{(1)})}{Q'(c_j^{(1)})}>0.
$$
\end{proof}

This property  is inherited by an arbitrary combination $P$ and $Q_\vk$, $\vk\in\bbR$, in particular by
 $P$ and $Q^{(t)}$.


We are in a position to prove the main result of this section.

\begin{theorem}\label{thmain4}
Let $\{\eta_k\}_{k=1}^n$ and $\{\eta^{(t)}_k\}_{k= 1}^n$ be the space and time generating frequencies in the KdV evolution, respectively. Then, arbitrary variations  $\{\dot\eta_k\}_{k=1}^n\in\bbR^n$ and $\{\dot\eta^{(t)}_k\}_{k=1}^n\in\bbR^n$ can be obtained by unique variations of the ramification points of $\cR$, i.e., of $\{\dot a_k,\dot b_k\}_{k=1}^n$ with the fixed additional ramification points zero and infinity.
\end{theorem}

\begin{remark}
In particular, in the above theorem one can set $\dot\eta_k=0$ and choose an arbitrary $\dot \eta_k^{(t)}$.
\end{remark}

\begin{proof}[Proof of Theorem \ref{thmain4}]
We have to show that the linear system \eqref{822sept21}--\eqref{622sept21} is uniquely solvable.
By $\cX$ we denote the matrix of this system.
Assume that the matrix is degenerate, then there exists a nontrivial vector $(x,y)\in \bbR^{2n}$ such that
$$
(x,y)\cX=0,\quad x=\begin{pmatrix} x_1&\dots& x_n\end{pmatrix},\quad
y=\begin{pmatrix} y_1&\dots& y_n\end{pmatrix}.
$$
Using the special form of the matrix $\cX$, we get
\begin{eqnarray}
O^{x}(a_j)P(a_j)+O^{y}(a_j)Q^{(t)}(a_j)=0 \nonumber\\
O^{x}(b_j)P(b_j)+O^{y}(b_j)Q^{(t)}(b_j)=0,\label{915sept21}
\end{eqnarray}
where
$$
O^x(z)=\sum_{k=1}^n x_k O^{(k)}(z),\quad O^y(z)=\sum_{k=1}^n y_k O^{(k)}
(z)
$$
are polynomials of degree at most $n-1$.

Let $T_1(z)=\prod_{j=1}^n(z-a_j)(z-b_j)$, that is, $T(z)=zT_1(z)$.
Since  $O^x(z)P(z)+O^y(z)Q^{(t)}(z)$ is a polynomial of degree at most $n$,
\eqref{915sept21} means that it is collinear to $T_1(z)$,
\begin{equation}\label{1015sept21}
O^x(z)P(z)+O^y(z)Q^{(t)}(z)=C T_1(z).
\end{equation}
If $C=0$, we have $(x,y)=0$ and the theorem is proved. Alternatively, we can normalize $(x,y)$ such that $C=1$. Consequently $O^y$ should be a monic polynomial in this normalization.

Denote by $\{c^{\vk}_j\}_{j=1}^{n+1}$ the zeros of $Q_\vk(z)$. Recall that they are interlacing with $\{c_j\}_{j=1}^n$, the zeros of $P$. In particular,  $\{c^{(t)}_j\}_{j=1}^{n+1}$ are the zeros of $Q^{(t)}(z)$ and $\{c^{(1)}_j\}_{j=1}^{n+1}$, $c_{n+1}^{(1)}=0$, are the zeros of $Q(z)$.

By \eqref{1015sept21} we can restore $O^x$ and $O^y$ using the interpolation formulas
\begin{align}\nonumber
O^{x}(z)=\sum_{j=1}^{n+1}\frac{Q^{(t)}(z)}{(Q^{(t)})'(c^{(t)}_j)(z-c^{(t)}_j)}\frac{T_1(c^{(t)}_j)}{P(c^{(t)}_j)}\\
\label{1115sept21}
O^{y}(z)=\sum_{j=1}^n\frac{P(z)}{P'(c_j)(z-c_j)}\frac{T_1(c_j)}{Q^{(t)}(c_j)}.
\end{align}
The polynomial $O^y$ is automatically of degree $n-1$, but $\deg O^x=n-1$ is equivalent to the statement that $O^y$ is monic, see \eqref{1015sept21} with $C=1$. Thus, if $(x,y)\not=0$, then
by \eqref{1115sept21},
\begin{equation}\label{323sept21}
1=\sum_{j=1}^n\frac{T_1(c_j)}{P'(c_j)Q^{(t)}(c_j)}=\sum_{j=1}^n\frac{T_1(c_j)}{P'(c_j)Q(c_j)}.
\end{equation}

To finalize the proof,  we will show that in fact
\begin{equation}\label{423sept21}
1<\sum_{j=1}^n\frac{T_1(c_j)}{P'(c_j)Q(c_j)}.
\end{equation}
For a fixed $\vk\in\bbR$, consider
$$
\frac{T_1(z)}{P(z)Q_\vk(z)}=\sum_{j=1}^n \frac{1}{z-c_j}\frac{T_1(c_j)}{P'(c_j)Q_\vk(c_j)}
+\sum_{j=1}^{n+1} \frac{1}{z-c^{\vk}_j}\frac{T_1(c^\vk_j)}{P(c^\vk_j)Q'_{\vk}(c^\vk_j)}.
$$
Multiplying by $z$ and passing to the limit $z\to\infty$, we get
$$
1=\sum_{j=1}^n \frac{T_1(c_j)}{P'(c_j)Q(c_j)}
+\sum_{j=1}^{n+1} \frac{T_1(c^\vk_j)}{P(c^\vk_j)Q'_{\vk}(c^\vk_j)}.
$$
Here, like in \eqref{323sept21}, we take into account that $Q_\vk(c_j)=Q(c_j)$. We set now
$$
\vk=\frac{Q(a_n)}{P(a_n)}.
$$
In this case $c^\vk_{n+1}=a_n$ and we have
\begin{equation}\label{623sept21}
1-\sum_{j=1}^n \frac{T_1(c_j)}{P'(c_j)Q(c_j)}=
\sum_{j=1}^{n} \frac{T_1(c^\vk_j)}{P(c^\vk_j)Q'_{\vk}(c^\vk_j)},
\end{equation}
that is, in the RHS the last term, dealing with $j=n+1$, is vanishing.
We point out, see the proof of Lemma \ref{l120sept21}, that each open gap
$(a_j,b_j)$ contains one of zeros of $Q_\vk$. That is, with a suitable enumeration,
$c^\vk_j\in(a_j,b_j)$, $j=1,\dots,n$. By the interlacing property,
$P(c^\vk_j)Q'_{\vk}(c^\vk_j)>0$ and in the gaps $T_1(\l)<0$, for all $\l\in(a_j,b_j)$.
Thus the RHS in \eqref{623sept21} is a sum of negative numbers, and \eqref{423sept21} and simultaneously the theorem are proved.
\end{proof}

\section{The Abel Map. Conjecture \ref{th13sept21}}\label{sect5}

\subsection{Parametrization of Reflectionless Operators by Divisors}
Let
\begin{equation}\label{130sept21}
\cM(z)=-\begin{pmatrix}
-1/m_-& 1\\ 1 &m_+
\end{pmatrix}^{-1}=\frac{1}{m_++m_-}\begin{pmatrix}
m_+m_-&- m_-\\ -m_- &-1
\end{pmatrix},
\end{equation}
where $m_\pm$ were defined in Section \ref{subs1_3}.
This is the Weyl $\cM$-matrix: it is a matrix Herglotz function, with an integral representation
\begin{equation}\label{matrixHerglotzRep}
\cM(z) = B + \int \left( \frac 1{\lambda - z} - \frac{\lambda}{\lambda^2 +1} \right) d\rho(\lambda),
\end{equation}
where $B = B^*$ and $\rho$ is a $2\times 2$ positive matrix-valued measure on $\bbR$ ($\cM$ has no point mass at $\infty$, due to the leading asymptotics of $m_\pm$). $L_V$ is unitarily equivalent to the operator of multiplication by $\lambda$ in $L^2(\bbR, \bbC^2, d\rho(\lambda))$. In particular, $\tr \cM$ is a Herglotz function whose corresponding measure on $\bbR$ is a maximal spectral measure for $L_V$.

We use special notations for diagonal entries of the matrix function $\cM(z)$,
\[
R = -  \frac 1{m_- + m_+},\quad R_1 =  \frac {m_+ m_-}{m_- + m_+}.
\]
They are Herglotz class functions in $\bbC_+$. With the symmetry $\overline{R(\bar z)} = R(z)$,
$\overline{R_1(\bar z)} = R_1(z)$ they are well defined in $\Omega$.
For a continuous potential $V$, the following asymptotic relation holds,
\begin{equation}\label{Rinfinity2}
R(z) = \frac 1{2 \sqrt{-z}} \left(1+ \frac{V(0)}{2z} + o(\lvert z \rvert^{-1}) \right), \qquad z \to -\infty.
\end{equation}

The map
$$
\mathcal{B} : \cR(\E) \to \cD(\E), \; V \mapsto D
$$
is defined by the following construction.

For any $V \in \cR(\E)$, $R(z)$ is  the full-line resolvent function.
It can be restored by its arguments on the real axis (up to a positive multiplier). Due to \eqref{reflectionless}, $R(z)$
takes  purely imaginary boundary values on $\E$. In gaps it is real valued. Since $R$ is holomorphic and strictly increasing on gaps, there exist $\lambda_j \in [a_j, b_j]$ such that $R < 0$ on $(a_j, \lambda_j)$ and $R > 0$ on $(\lambda_j, b_j)$. Note that in particular, we have $\l_j=b_j$ if $R<0$ in the gap and $\l_j=a_j$ if $R>0$ in it. Finally, since $\sigma(L_V)\subset \bbR_+$, $R(\l)>0$ on the negative half axis. Thus the argument of the function is completely defined (for a.e. $\l\in\bbR$). Having in mind the leading term in the asymptotic
\eqref{Rinfinity2}, we get
\begin{equation}\label{resolventproduct}
R(z) = \frac 1{2\sqrt{-z}} \prod_{j=1}^\infty  \frac{(z-\lambda_j)}{\sqrt{(z-a_j)(z-b_j)}}.
\end{equation}

Note that $R$ and $R_1$ are holomorphic functions in $\Omega=\bbC \setminus \E$. Since
$$
m_++m_-=-\frac{1}{R} \quad\text{and}\quad R_1=-R m_+m_-
$$
 if $\lambda_j \in (a_j, b_j)$, then $\lambda_j$ is a pole of exactly one of the functions $m_\pm$, which determines the coordinate $\epsilon_j \in \{ - 1, +1\}$.

In summary, we get the map $\cB$ as follows: to any reflectionless $V$ corresponds the Dirichlet data $D = \{ (\lambda_j, \epsilon_j) \}_{j=1}^\infty \in \prod_{j=1}^\infty I_j$ with each $I_j$ a topological circle, see Definition \ref{def127}.

Our specific constraints \eqref{setE1}, \eqref{setE2} on the set $\sE$ allow us to show that the divisor uniquely determines the $m$-functions (and therefore the potential $V$).

\begin{lemma}
If the set $\E$ is of the form \eqref{setE1}, \eqref{setE2}, the functions $m_\pm$ are uniquely determined by $D= \{ (\lambda_j, \epsilon_j) \}_{j=1}^\infty \in \cD(\E)$ by
\begin{equation}\label{additivesplitmpm}
m_\pm(z) = - \frac 1{2 R_D(z)} \pm \sum_{j: \lambda_j \in (a_j,b_j)} \frac{\epsilon_j}{2R_D'(\lambda_j)} \frac{1 }{\lambda_j - z},
\end{equation}
where $R_D=R$ is given by \eqref{resolventproduct}.
\end{lemma}

\begin{proof}
Since $m_+ + m_- = - 1/R$, this is a question of splitting the Herglotz function $-1/R$ into two, which can be done by splitting the measure $\sigma$ in its Herglotz representation,
\[
- \frac 1{R(z)} =  \beta + \int \left( \frac 1{\lambda - z} - \frac{\lambda}{1+\lambda^2} \right) d\sigma(\lambda)
\]
(by \eqref{resolventproduct}, there is no point mass at $\infty$).
Note that a priori the origin is the only possible support for the singular component of the measure $\sigma$ on $\sE$.

On $\bbR \setminus \E$, the measure $\sigma$ can only have isolated point masses. The function $-1/R$ has simple poles precisely at $\lambda_j \in (a_j, b_j)$ with residues $\sigma(\{\lambda_j\}) = 1 / R'(\lambda_j)$. Note that
\begin{equation}\label{7jan1}
\sum_{j: \lambda_j \in (a_j, b_j)} \frac 1{R'(\lambda_j)} \le \sigma([0,b_1]) < \infty.
\end{equation}
For any interval $[c,d] \subset (b_{k+1}, a_k)$, by \eqref{resolventproduct}, $-1/R$ is uniformly bounded on the rectangle $[c, d] \times (0,1]$, so $\sigma$ is purely absolutely continuous there. At any endpoint $x_* \in \{ a_k, b_k \mid k\in \bbN\}$, by \eqref{resolventproduct}, $R(z) = O( \lvert z - x_* \rvert^{-1/2})$ as $z \to x_*$, so $\sigma(\{x_*\}) =0$. Finally, since $R$ is strictly positive and increasing on $(-\infty,0)$,
\[
\lim_{\lambda \uparrow 0} ( -  1/ R(\lambda)) = \sup_{\lambda \in (-\infty, 0)} ( -  1/ R(\lambda)) \in (-\infty,0],
\]
so
\[
\sigma(\{0\}) = \lim_{\lambda \uparrow 0} \left( -  \frac{\lambda}{R(\lambda)} \right) = 0.
\]

By the reflectionless condition \eqref{reflectionless}, the absolutely continuous part of $\sigma$ must be split equally between $m_\pm$, and the poles of $-1/R$ at $\lambda_j \in (a_j, b_j)$ must go completely to $m_{\epsilon_j}$. Since $\sigma$ has no singular part on $\E$, this describes the split up to an additive real constant:
\begin{equation}\label{additivesplitmpm2}
m_\pm(z) = \pm c - \frac 1{2 R_D(z)} \pm \sum_{j: \lambda_j \in (a_j,b_j)} \frac{\epsilon_j}{2R_D'(\lambda_j)} \frac{1 }{\lambda_j - z}
\end{equation}
for some $c \in \bbR$. Note that
\[
m_+(z) - m_-(z) = 2 c +  \sum_{j: \lambda_j \in (a_j,b_j)} \frac{\epsilon_j}{R_D'(\lambda_j)} \frac{1 }{\lambda_j - z} \to 2c
\]
as $z \to -\infty$, by dominated convergence justified due to \eqref{7jan1}. However, $m_\pm(z) = - \sqrt{-z} + o(1)$ as $z \to -\infty$ by a result of Atkinson, which implies $c = 0$ and proves \eqref{additivesplitmpm}. If we assume from the start that $V$ is continuous, we can replace this by the more precise
\begin{equation}\label{7-2feb21}
m_\pm(z) = - \sqrt{-z} - \frac{V(0)}{2 \sqrt{-z}} + o(\lvert z\rvert^{-1/2}).
\end{equation}
\end{proof}

Thus, the divisor $D$ uniquely determines the Weyl functions $m_\pm$, and by the Borg--Marchenko theorem, they uniquely determine the restrictions of $V$ to $(0,\pm \infty)$.  Conversely, for any divisor $D$, the formulas \eqref{additivesplitmpm} determine a reflectionless pair $m_\pm$ such that $\cM$ given by \eqref{matrixHerglotzRep} has spectrum $\E$. Since $\E$ contains a half-line, this is in the Kotani--Marchenko class, and there exists a potential $V$ corresponding to the reflectionless pair $m_\pm$. Thus, for a divisor $D$, let us denote by $V_D$ the unique potential in $\cR(\E)$ corresponding to the divisor $D$.

Note that the leading behaviors of $m_\pm$ imply
\[
R_1 = - \frac 12 \sqrt{-z} - \frac{V(0)}{4\sqrt{-z}} + o(|z|^{-1/2}), \qquad z \to -\infty.
\]
In particular, $R_1(z) \to -\infty$ as $z \to -\infty$.

\subsection{The Abel Map and Canonical Products.  Weyl-Titchmarsh Functions in Terms of Canonical Products}

The \emph{complex Green function} $\Phi_{z_0}$ is defined by
\[
\lvert \Phi_{z_0}(z) \rvert = e^{-G_\Omega(z,z_0)}, \qquad \Phi_{z_0}(-1) > 0.
\]
This uniquely determines the character-automorphic function $\Phi_{z_0}$.
Due to the regularity of the domain $\Omega=\bbC\setminus \sE$ and symmetry, this phase normalization gives, for $\l \in [0,\infty) \setminus \E$,
\begin{equation}\label{14jan1}
\lim_{z\to -\infty}\Phi_{\l}(z)=1.
\end{equation}
Note that on the universal covering $\Phi_{z_0}$ represents a Blaschke product and vice versa a Blaschke product constructed with zeros along the orbit $\{\g(\z_0)\}_{\g\in\G}$, $\z_0\in\bbD$, is a complex Green function w.r.t. $z_0=\L(\z_0)$.

As before, $\om(\sF,z_0)$ denotes the harmonic measure of the set $\sF\subset\sE$ computed at $z_0\in\Omega$.
For $\l_j\in(a_j,b_j)$, $\Phi_{\l_j}(z)$ can be represented by means of a conformal mapping on the following comb domain.
\begin{proposition}[see e.g. \cite{SY97}]
Let $\omega_k^{\l_j}=\om(\sE_k,\l_j)$. There exists a system of heights $\{h^{\l_j}_k\}_{k\ge 0,k\not=j}$ such that $\Phi_{\l_j}(z)=e^{i\theta_{\l_j}(z)}$, $z\in\bbC_+$, and $\theta_{\l_j}$ is the conformal mapping onto the half-strip
$$
\{w=u+iv, v>0, u\in(\pi(\om_j^{\l_j}-1),\pi\om_j^{\l_j})\}
$$
with the system of slits
$$
\cup_{k>j}\{\pi\om_k^{\l_j}+iv,\ v\le h^{\l_j}_k\}\cup_{k<j}\{\pi(\om_k^{\l_j}-1)+iv,\ v\le h^{\l_j}_k\}\cup
\{iv,\ v\le h^{\l_j}_0\}.
$$
In this conformal mapping, the last slit corresponds to the negative half-axis, that is, $\theta_{\l_j}(\infty)=-0$ and $\theta_{\l_j}(0)=+0$; and the third normalization condition is $\theta_{\l_j}(\l_j)=\infty$.
\end{proposition}

According to this statement for the generators $\{\g_k\}$ of $\pi_1(\Omega)$, we have
\begin{equation}\label{128sept21}
\Phi_{\l_j}\circ\g_k=e^{2\pi i\om_k^{\l_j}} \Phi_{\l_j}.
\end{equation}
Respectively, an alternative definition (to Definition \ref{def227})  of the Abel map can be given in terms of characters corresponding to the following canonical products.

\begin{lemma}
For an arbitrary divisor $D \in \cD(\E)$, we define the canonical product $L_D$ as follows,
\begin{equation}\label{11jan211}
L_D(z) = \left\{ \prod_{j\ge 1}\frac{(z-\l_j)}{(z-b_j)\Phi_{\l_j}(z)}\right\}^{\frac 1 2}\prod_{j\ge 1}\Phi_{\l_j}(z)^{\frac{1+\e_j}2},
\end{equation}
where we take the branches of square roots so that $L_D > 0$ on $(-\infty, 0)$.
Then the character of $L_D$ is $\cA(D)$.
\end{lemma}

\begin{proof}
The convergence of the product as well as of each sum in the definition
\eqref{12sept21} follows from the Widom condition \eqref{11sept21}.  It remains, using \eqref{128sept21}, to compute the character of each factor in the canonical product.
\end{proof}

Our goal in this subsection is to get a representation for the Weyl-Titchmarsh functions in terms of canonical products, see Theorem \ref{th5828}. Essentially, its proof is based on the following theorem.

\begin{theorem}[ {\cite[Theorem D]{SY97}} ]\label{th31}
Let $\E$ be a Dirichlet regular Widom set and $f$ a meromorphic function on $\Omega$ with
$\Im f(z)\ge 0$ for $z\in\bbC_+$ and $\overline{f(\bar z)} = f(z)$. If the poles of $f$ satisfy the condition
\[
\sum_{\l: f(\l) = \infty} G(\l, -1) < \infty,
\]
then $f$ is of bounded characteristic and its inner factor is a quotient of Blaschke products (i.e., $f$ has no singular inner factor).
\end{theorem}

Note that due to this theorem, \eqref{11jan211} is the outer-inner decomposition for the product $L_D$.

According to the previous subsection, $D\in\cD(\sE)$ one-to-one corresponds to $V\in\cR(\sE)$. Respectively, the associated Weyl-Titchmarsh functions depend on $D$, compare \eqref{additivesplitmpm}, but
to simplify notation for a moment we drop the index $D$ and write $m_\pm$.

Due to monotonicity, $\lim_{\l\to-0}m_+(\l)$ exists and we denote this value by $m_+(0)$. Let
\begin{equation}\label{9-2feb21}
n_+(z)=m_+(z)-m_+(0).
\end{equation}
Note that $n_+(\l)\le 0$ on the negative half-axis. We define $n_-(z)$ via the reflectionless property:
$$
n_-(\l)=-\overline{n_+(\l)}=m_{-}(\l)+m_+(0), \quad  \l\in\sE.
$$
Note that as before
\begin{equation}\label{dd1-4feb21}
n_+(z)+n_-(z)=-\frac{1}{R(z)}.
\end{equation}
In particular this means that $n_-(0)$, which we define via the limit on $\bbR_-$, is negative.
 By monotonicity, $n_-(\l)\le 0$ for $\l\in\bbR_-$. So, the adjoint second function is
\begin{equation}\label{0-2feb21}
\tilde R_1(z)=\frac{n_+(z)n_-(z)}{n_+(z)+n_-(z)}=\frac{(m_+(z)-m_+(0))(m_-(z)+m_+(0))}{m_+(z)+m_-(z)}.
\end{equation}
Due to our construction, $\tilde R_1(\l)\le 0$ for $\l\in\bbR_-$. Similarly to \eqref{resolventproduct}, there is a collection
$\{\l_j^{(1)}\}_{j\ge 1}$, $\l_j\in[a_j,b_j]$ such that
\begin{equation}\label{dd2-4feb21}
\tilde R_1(z)=-\sqrt{-z}\prod_{j\ge 1}\frac{z-\l_j^{(1)}}{\sqrt{(z-a_j)(z-b_j)}}.
\end{equation}
Moreover we can introduce a divisor $D_1=\{(\l_j^{(1)},\e_j^{1})\}$ such that $\e_j^{1}=\pm 1$ if and only if $\l^{(1)}_j$ is a zero of $n_{\pm}$.

\begin{lemma}\label{lem51}
In terms of $D$ and $D_1$ defined above, the intermediate function $n_+$ admits the following representation,
\begin{equation}\label{1-2feb21}
n_+(z)=-\sqrt{-z}\frac{L_{D_1}(z)}{L_D(z)}.
\end{equation}
\end{lemma}

\begin{proof}
Due to \eqref{resolventproduct}, \eqref{dd1-4feb21}, \eqref{0-2feb21}, and \eqref{dd2-4feb21}, we have
$$
n_+(z) n_-(z)=-z\prod_{j\ge 1}\frac{z-\l_j^{(1)}}{z-\l_j}.
$$
We can use  Theorem~\ref{th31}. Since $|n_+(\l)|=|n_-(\l)|$, $\l\in\sE$, the outer parts of both functions coincide. Since the inner parts are the ratio of Blaschke products, they can be recovered up to unimodular constants from $\e_j$ and $\e_j^{(1)}$. The unimodular constants are uniquely defined by the signs of $n_\pm$ on the negative half-axis. As a result we get \eqref{1-2feb21}.
\end{proof}

The remaining parameter $m_+(0)$ can be found from the asymptotics at infinity. Let us mention that all involved function in fact are holomorphic (or at most have a pole) at infinity with respect to the local parameter $\tau$, $\tau^2=1/z$.

\begin{lemma} Let
\begin{equation}\label{5-1feb21}
\cQ_1(D) = \sum_j \left(\frac{a_j + b_j} 2 -  \lambda_j\right)
\end{equation}
and
\begin{equation}\label{6-1feb21}
\cQ_2(D)=\sum_{j\ge 1}M(\l_j)\e_j, \quad M(z)=\Im \Theta(z).
\end{equation}
Then
\begin{equation}\label{e.RDasymptotics}
R_D(z) = \frac 1{2\sqrt{-z}} \left( 1 + \frac{\cQ_1(D)}{z} + O \Big( \frac 1{z^2} \Big) \right), \qquad z \to -\infty
\end{equation}
and
\begin{equation}\label{e.LDoverPhic}
L_D(z) = 1 - \frac{\cQ_2(D)}{\sqrt{-z}} +\frac{\cQ_1(D)-\cQ_1(D_0) - \cQ_2(D)^2}{2z} + O(z^{-3/2}).
\end{equation}
Recall $D_0=\{b_j\}_{j\ge 1}$.
\end{lemma}
\begin{proof}
The first relation \eqref{e.RDasymptotics} is well known.

For an arbitrary $\l_*\in\bbR_+\setminus\sE$,
 passing to the local coordinate $z=1/\tau^2$ we have an analytic function
$$
\phi(\tau)=\Phi_{\l_*}(1/\tau^2)
$$
in a certain vicinity of the origin. Due to our normalization $\phi(0)=1$, therefore $\log \phi(\tau)$ is  well defined and possesses a power series expansion
\begin{equation}\label{1feb21}
\log \phi(\e)=C_1\tau+C_2\tau^2+C_3\tau^3+\dots .
\end{equation}
Finally, we note that $\Phi(z)$ is real valued on the negative half-axis and
$$
|\Phi(\l\pm i0)|=1 \quad\text{for}\quad \l>b_1.
$$
This implies the symmetry property  $\phi(-\tau)=1/\phi(\tau)$. Consequently, in \eqref{1feb21} all even coefficients vanish, $C_{2n}=0$. In particular, we get
$$
\Phi_{\l_*}(z)=1+\frac{C_1}{\sqrt{z}}+ \frac{C_1^2}{2z}+O\left(\frac{1}{\sqrt{z}^3}\right),\quad z\to\infty.
$$

Now we will obtain a representation for $C_1$ in terms of special functions associated with the domain $\Omega$. For $\lambda_* \in (a_j, b_j)$ and $\l\in\bbR_-$, we have
\[
G_\Omega(\l, \l_*) = \Re \int_{a_j}^{\lambda_*} \frac 1{t - \l} \frac{R_{D_{c(\lambda)}}(t)}{R_{D_{c(\lambda)}}(\lambda)} \, dt.
\]
with a suitable divisor $D_{c(\l)}=\{(c_j({\l}),1)\}_{j\ge 1}$. We will use that uniformly
\[
(t - \lambda)
R_{D_{c(\l)}}
(\l) = \frac{\sqrt{-\lambda}}2 + o(\dots), \qquad \lambda \to -\infty.
\]
Since $D_{c(\l)}\to D_c$, where $D_c$ is the divisor of critical points of the Martin function $\Im\Theta$, we get the asymptotics
\[
G_\Omega(\l, \l_*) = \frac 2{\sqrt{-\l}} \int_{a_j}^{\l_*} R_{D_c}(t) dt + o(\dots) = \frac 2{\sqrt{-\lambda}} M(\lambda_*) + o(\dots)    \qquad  \l \to -\infty.
\]

By the Widom condition, an arbitrary product
$$
\Phi(z)=\prod_{j\in I}\Phi_{\l_j}(z)
$$
is convergent for an arbitrary collection of indexes $I$. Moreover an arbitrary open arc $\Lambda^{-1}((b_1,\infty])$ is free of accumulation points of zeros of this product. It is well known, see e.g. \cite[Chap. II, Theorem 6.1]{Gar}, that $\Phi(\Lambda(\z))$  is analytic on all these arcs.
Let
$$
\Phi_D(z)=\prod_{j\ge 1}\Phi_{\l_j}(z)^{\e_j}=\frac{\prod_{j:\e_j=1} \Phi_{\l_j}(z)}{\prod_{j:\e_j=-1} \Phi_{\l_j}(z)}.
$$
Then for this ratio of Blaschke products we get
$$
\Phi_{D}(z)=1+\frac{C_1(D)}{\sqrt{z}}+ \frac{C_1(D)^2}{2z}+O\left(\frac{1}{\sqrt{z}^3}\right),\quad z\to\infty
$$
with
$$
C_1(D)=2i \cQ_2(D).
$$

Finally, we have
\begin{align*}
& {L_D(z)}
 = \left\{\prod_{j\ge 1}\frac{z-\l_j}{z-b_j} \Phi_D(z)\right\}^{\frac 1 2} \\
& = \left( 1+\frac{\cQ_1(D)-\cQ_1(D_0)}{2z}+ O(z^{-2}) \right)
\left( 1 - \frac{2 \cQ_2(D)}{\sqrt{-z}} - \frac{4 \cQ_2(D)^2}{2z} + O(z^{-3/2}) \right)^{1/2} \\
& = \left( 1+\frac{\cQ_1(D)-\cQ_1(D_0)}{2z}+ O(z^{-2}) \right) \left( 1 - \frac{\cQ_2(D)}{\sqrt{-z}} - \frac{\cQ_2(D)^2}{2z} + O(z^{-3/2}) \right) \\
& = 1 - \frac{\cQ_2(D)}{\sqrt{-z}} +\frac{\cQ_1(D)-\cQ_1(D_0) - \cQ_2(D)^2}{2z} + O(z^{-3/2})
,
\end{align*}
which establishes \eqref{e.LDoverPhic}.
\end{proof}

\begin{theorem}\label{th5828}
Let $m_+=(m_D)_+$ be the Weyl-Titchmarsh function defined by a divisor $D$. Let the divisor $D_1$ be associated to this function by \eqref{1-2feb21}. Then
\begin{equation}\label{3-2feb21}
m_+(z)=\cQ_2(D)-\cQ_2(D_1)-\sqrt{-z}\frac{L_{D_1}(z)}{L_D(z)}.
\end{equation}
\end{theorem}

\begin{proof}
Recall that $m_+(z)$ is fixed by its asymptotics at infinity \eqref{7-2feb21}. Using \eqref{e.LDoverPhic} and
\eqref{1-2feb21}  we get a correction constant $m_+(0)$ in \eqref{9-2feb21} and respectively have  \eqref{3-2feb21}.
\end{proof}

\subsection{The Hardy Subspace $\cH^2_D$ Generated by $(m_D)_+$}

The Hardy space $\hat H^2(\a)$ is defined as the collection of character-automorphic functions $f$ with the character $\a$ such that  $|f(z)|^2$ possesses a harmonic majorant, i.e.,
$$
 |f(z)|^2\le u(z)
$$
with a certain harmonic function $u(z)$ in $\Omega$.

For Widom domains an equivalent definition is the following:
$f\in \hat H^2(\a)$ if
\begin{itemize}
\item[a)] $f$ is character automorphic,  $f\circ\g=e^{2\pi i\a(\g)} f$, $\g\in\pi_1(\Omega)$,
\item[b)] $f$ is of Smirnov class, $f\in \cN_+(\Omega)$,
\item[c)] the boundary values of $f$ are square-integrable with respect to the harmonic measure in the domain,
\begin{equation}\label{229sept21}
\lVert f \rVert_{\hat H^2(\a)}^2 = \int_\E ( \lvert f(\l+i0) \rvert^2 + \lvert f(\l - i0)\rvert^2 ) \, \omega(d\l,-1) < \infty.
\end{equation}
\end{itemize}

However, one can substitute the harmonic measure in \eqref{229sept21} by an arbitrary measure $\sigma$ that is of Szeg\"o class in the following sense.

\begin{proposition}
Let $d\sigma=|\vp|^2 d\om(d\l, -1)$, where $\vp$ is an outer character automorphic function with the character $\a_\vp$. If $f$ obeys a), b) and
\begin{equation}\label{329sept21}
 \int_\E ( \lvert f(\l+i0) \rvert^2 + \lvert f(\l - i0)\rvert^2 ) \, d\sigma < \infty,
\end{equation}
then $\vp f\in \hat H^2(\a+\a_\vp)$. Conversely, if we define
$$
\hat H^2_\sigma(\a)=\{f\in \cN_+(\Omega),\ f\circ\g=e^{2\pi i\a(\g)} f, \ \text{\eqref{329sept21} holds}\},
$$
then $f\in \hat H^2_\sigma(\a)$ implies $f/\vp\in \hat H^2(\a-\a_\vp)$ and
$$
\|f/\vp\|_{\hat H^2(\a-\a_\vp)}=\|f\|_{\hat H_\s^2(\a)}
$$
with the norm defined by the integral \eqref{329sept21}.
\end{proposition}

Since the divisor $D_0=\{b_j\}_{j\ge 1}$ plays a special role in our construction, we will work with the following system of Hardy spaces.

\begin{definition}
We use the special notation $\hat\cH^2(\a)$ for the Hardy space $\hat H^2_{\s}$ with the following measure,
\begin{equation}\label{230sept21}
2\pi \, d\sigma=|R_{D_0}(\l)| \, d\l=\sqrt{\prod_{j\ge 1}\frac{\l-b_j}{\l-a_j}}\frac{d\l}{2\sqrt{\l}},\quad \l\in\sE.
\end{equation}
Theorem \ref{th31} implies that this measure belongs to the Szeg\"o class.
\end{definition}

We will associate a Hardy space $\cH^2_D$ to an arbitrary divisor $D$ such that $D\not\in\cD_b(\sE)$. We note that the collection $\cD_b(\sE)$, see \eqref{129sept21}, describes all divisors $D$ for which $R_D$ contains a singular component in its integral representation. In other words, the associated 1-D  Schr\"odinger operator has $0$ as an eigenvalue.

\begin{lemma}
Assume that $\Omega$ is a regular domain of Widom type such that $\sE$ obeys \eqref{33sept21}. The associated Schr\"odinger operator $L_{V_D}$ has $0$ as an eigenvalue if and only if $D \in \cD_b(\sE)$.
\end{lemma}

\begin{proof}
The spectral properties of $L_D$ are encoded by the matrix measure \eqref{matrixHerglotzRep}. They in turn are given by measures associated to $R_D$ and $(R_D)_1$. Due to our permanent assumption \eqref{setE2}, the origin is the only possible support of the singular component of these measures. Since $(R_D)_1$ is bounded in the origin, the only possible condition to have an eigenvalue is given by
$$
\lim_{z\to 0}(-z)R_{D}(z)>0.
$$
This in turn is equivalent to \eqref{53sept21}.
\end{proof}

As a result we get that as soon as $D \not\in \cD_b(\sE)$, the measure associated to $R_D$ is absolutely continuous.

\begin{theorem}
Let $D \not\in \cD_b(\sE)$. Then the relation
\begin{equation}\label{1-4feb21}
k_D(z,z_0)=L_D(z)\frac{(m_D)_+(z)-\overline{(m_D)_+(z_0)}}{z-\bar z_0}\overline{L_D(z_0)}
\end{equation}
defines a reproducing kernel of a subspace of $\hat \cH^2(\a)$ with $\a=\cA(D)$.
\end{theorem}

\begin{proof}
Since $D$ is fixed we will leave this index implicit in $m_\pm(z)$. We use the following simple trick,
\begin{equation}\label{22jan1}
\frac{-\cM^{-1}(z)+(\cM(z)^{-1})^*}{z-\bar z}=\cM^{-1}(z)\frac{\cM(z)-\cM(z)^*}{z-\bar z}(\cM(z)^{-1})^*,
\end{equation}
where $\cM$ is given by \eqref{130sept21}, This allows us to represent
$$
\frac{m_+(z)-\overline{m_+(z)}}{z-\bar z}
$$
via an integral over $\sE$. Since $D \not\in D_b(\sE)$, the associated  matrix measure $\rho$ (cf.~\eqref{matrixHerglotzRep}) is absolutely continuous and we have
$$
\frac{\cM(z)-\cM(z)^*}{z-\bar z}=\int_{\sE}\frac{d\rho(\l)}{|\l-z|^2},\quad d\rho(\l)=\frac 1 {2\pi i}(\cM(\l)-\cM(\l)^*) \, d\l.
$$
Using reflectionlessness on $\sE$, we find
\begin{align*}
\frac{m_+(\l)+{m_-(\l)}} i 2\pi\rho'(\l) & =
\begin{pmatrix}
2|m_+(\l)|^2&-m_+(\l)-\overline{m_+(\l)}\\
-m_+(\l)-\overline{m_+(\l)}
&2
\end{pmatrix}
\\
& = \begin{pmatrix}
\overline{m_+(\l)}&-1\\
m_+(\l)
&-1
\end{pmatrix}^*\begin{pmatrix}
\overline{m_+(\l)}&-1\\
m_+(\l)
&-1
\end{pmatrix}.
\end{align*}
Thus,
\begin{align*}
\frac{m_+(z)-\overline{m_+(z)}}{z-\bar z}  =
\begin{pmatrix}
1&m_+(z)
\end{pmatrix}\int_{\sE}\frac{d\rho(\l)}{|\l-z|^2}
\begin{pmatrix}
1\\ \overline{m_+(z)}
\end{pmatrix}
\\
 = \int_{\sE}\frac{|m_+(\l)-m_+(z)|^2+|\overline{m_+(\l)}-m_+(z)|^2}
{|\l-z|^2}\frac{i \, d\l}{2\pi (m_+(\l)+m_-(\l))}.
\end{align*}

Now we note that
$$
\frac{m_+(\l)+m_-(\l)} i |L_D(\l)|^2= {2}{\sqrt \l}\prod_{j\ge 1}\frac{\sqrt{(\l-a_j)(\l-b_j)}}{\l-\l_j} \prod_{j\ge 1}\frac{\l-\l_j}{\l-b_j}
=\frac{1}{|R_{D_0}(\l)|}.
$$
Thus, we arrive at the identity
\begin{align*}
k_D(z,z) & = L_D(z) \frac{m_+(z)-\overline{m_+(z)}}{z-\bar z}\overline{L_D(z)} \\
& =\frac 1 {2\pi} \int_E \left( |k_D(\l+i0,z)|^2 + |k_D(\l-i0,z)|^2 \right) \, |R_{D_0}(\l)| \, d\l.
\end{align*}
In the last expression we used the following observation. Due to our normalization, the canonical products $L_D,  L_{D_1}$ are real on the negative half-axis. Therefore the limits of $L_D(z)$ from the top and bottom are related by $L_D(\l-i0) = \overline{L_D(\l+i0)}$.

From the very beginning we had a kernel $k_D(z,z_0)$, which is Hermit positive and holomorphic, consequently it generates a certain Hilbert space of holomorphic functions. Now we proved that the scalar product in this space can be treated as the scalar product in $L^2_{d\s}$ with the chosen measure \eqref{230sept21}. Since all functions $k_D(z,z_0)$ are of Smirnov class and character automorphic with the character $\a = \cA(D)$, this Hilbert space is a subspace of $\hat\cH^2(\a)$.
\end{proof}

Due to \eqref{3-2feb21} we have the following representation,
\begin{equation}\label{3-4feb21}
k_D(z,z_0) = \frac{-\sqrt{-z}L_{D_1}(z) \overline{L_D(z_0)} + \overline{\sqrt{-z_0}} L_{D}(z) \overline{L_{D_1}(z_0)}}{z-\bar z_0}.
\end{equation}
In what follows, the Hilbert space of character automorphic functions with the reproducing kernel $k_D$ is denoted $\cH^2_D$.

Let $D_* = \{(\l_j,-\e_j)\}_{j\ge 1}$ for $D=\{(\l_j,\e_j)\}_{j\ge 1}$.
By pseudocontinuation we get
\begin{equation}\label{3_m-2feb21}
(m_D)_-(z)=Q_2(D_*)-Q_2((D_1)_*)-\sqrt{-z}\frac{L_{(D_1)_*}(z)}{L_{D_*}(z)}.
\end{equation}

The following relation we call Wronskian identity.

\begin{lemma}\label{l.determinant}
Let $D_1$ be associated to $D$ as in Lemma~\ref{lem51}. Define
\begin{equation}\label{2-10feb21}
\cL_D(z)=\begin{pmatrix}
\sqrt{-z}L_{(D_1)_*}(z)& L_{D_*}(z)
\\
-\sqrt{-z}L_{D_1}(z)& L_D(z)
\end{pmatrix}.
\end{equation}
Then
\begin{equation}\label{330sept21}
\det\cL_D(z) = \frac 1{R_{D_0}(z)}.
\end{equation}
\end{lemma}

\begin{proof}
We have
\begin{align*}
\det\cL_D(z) & = L_D(z)L_{D_*}(z) \det
\begin{pmatrix}
-(n_D)_-(z)& 1\\
(n_D)_+(z)& 1
\end{pmatrix}  \\
& = \frac{L_D(z)L_{D_*}(z)}{R_D(z)}
 = \frac{\prod_{j\ge 1}\frac{z-\l_j}{z-b_j}}{\frac 1{2\sqrt{-z}} \prod_{j\ge1} \frac{(z-\lambda_j)}
 {
 \sqrt{(z-a_j)(z-b_j)}}}
 =\frac 1{R_{D_0}(z)}.
\end{align*}
Here we used \eqref{1-2feb21} in the first step, \eqref{dd1-4feb21} in the second step, \eqref{resolventproduct} and \eqref{11jan211} in the third step, and again \eqref{resolventproduct} with $D=D_0$ in the final step.
\end{proof}

We point out that due to this identity, the Wronskian determinant  does not depend on $D$.

\begin{proposition}\label{prop13sept21}
The composition map $\cA\circ\cB:\cR(\sE)\to\pi_1(\Omega)^*$ linearizes the space shift of a reflectionless potential,
that is,
$$
\cA(\cB(V(x+x_0))=\{\a_k-\eta_k x_0\mod 1\}_{k\ge 1}
$$
as soon as $\cB(V)\not\in\cD_b(\sE)$, in other words the spectrum of $L_V$ is absolutely continuous.
\end{proposition}

\begin{proof} In this setting, no modification is required for a
word by word repetition of the arguments given in \cite[Sect. 4]{DY}. Up to notations, our Proposition~\ref{prop13sept21} is the same as Proposition~4.8 from \cite{DY}.
\end{proof}

\section{The Set $\Xi$ of Singular Divisors. Theorem \ref{l13sept21}}

\subsection{The Easy Part}

In this subsection we show that $\Xi \supseteq \cA(\cD_b(\sE))$.
\begin{lemma}\label{theo53}
Let $\Omega$ be of Widom type and $\sE$ obey the condition
\eqref{33sept21}.
Let $\fj$ be the character generated by the function $\sqrt{-z}$.
Then the Abel map is not one-to-one. Moreover if $D\in \cD(\sE)$ obeys
\begin{equation}\label{11jan214}
\sum_{j\ge 1}\log\frac{\l_j}{a_j} < \infty,
\end{equation}
then the character $\cA(D)+\fj$ possesses a non-unique representation.
\end{lemma}

\begin{proof}
The conditions \eqref{33sept21} and \eqref{11jan214} imply that the limit value
$
R_D(0) < \infty,
$
where
\begin{align*}
R_D(z) & =\frac{1}{2\sqrt{b_1-z}}e^{\int_0^{b_1}\chi_D(\l)\frac{d\l}{\l-z}}, \\
\chi_D(\l) & = \begin{cases} 1/2,& \l\in \sE, \\
0, &\l\in(\l_j, b_j),\  j\ge 1, \\
1, &\l\in(a_j, \l_j),\  j\ge 1. \\
\end{cases}
\end{align*}
Therefore
$$
-(m_D)_-(0)-(m_D)_+(0)=\frac{1}{R_D(0)}>0,
$$
that is, $[(m_D)_+(0),-(m_D)_-(0)]$ is a proper interval. We fix an arbitrary internal $\mu$ from this interval. We point out that
$$
(n^\mu_D)_+(\l)=(m_D)_+(\l)-\mu<0,\quad (n^\mu_D)_-(\l)=(m_D)_-(\l)+\mu<0
$$
on the negative half axis, while $
(n^\mu_D)_+(z)+(n^\mu_D)_-(z)=-1/(R_D)(z)$. The same arguments as in the proof of Lemma  \ref{lem51} show that
$$
(n^\mu_D)_+=-\sqrt{-z}\frac{L_{D_1^\mu}(z)}{L_D(z)}
$$
with a certain $D_1^\mu\in \cD(\sE)$. Moreover, since this divisor  is defined via zeros of the functions $\{(m_D)_\pm(z)\mp\mu\}$, different $\mu$'s produce different divisors $D_1^\mu$. Since all such functions $(n^\mu_D)_+$ are single-valued in $\Omega$, all products $\{L_{D_1^\mu}\}$ correspond to the same character $\cA(D)+\fj$.
\end{proof}

\begin{corollary}
Assume that $\Omega$ is a regular domain of Widom type such that $\sE$ obeys \eqref{33sept21}.  
Then, $\Xi \supseteq \cA(\cD_b(\sE))$.
\end{corollary}

\begin{proof}
Using the notation from the proof of the previous lemma, note that
$D_1^\mu\in\cD_b(\sE)$ means that
$$
\lim_{z\to 0}(-z)R_{D_1^\mu}(z)>0.
$$
We have
$$
-\frac 1{R_{D_1^\mu}(z)}=\frac {-z}{(n_D^\mu)_+(z)}+\frac{-z}{(n_D^\mu)_-(z)}=\frac{z}{(n_D^\mu)_+(z)(n_D^\mu)_-(z)R_D(z)}.
$$
That is,
$$
R_{D_1^\mu}(z)=-\frac 1 z (n_D^\mu)_+(z)(n_D^\mu)_-(z)R_D(z).
$$
Since both $(n_D^\mu)_+(0)$ and $(n_D^\mu)_-(0)$ are negative, we have that $R_D(0)>0$.
The last expression means that for the given divisor $D$, \eqref{11jan214} holds.
\end{proof}

\subsection{Parametrization of the Defect Subspaces by Means of $J$-Contrac\-tive Matrix Functions}

We have proven that $\cA(\cD_b(\sE))\subset\Xi$. For $D\not\in\cD_b(\sE)$ we introduced the Hardy space
$\cH^2_D$ and proved  that $\cH^2_D\subset \hat\cH^2(\a)$, where $\cA(D)=\a$. As soon as $\a \in \Xi$, due to non-uniqueness, there should be $D$ corresponding to this $\a$ such that $\cH^2_D$ is a proper subset, see Corollary~\ref{c11oct21}.

\bigskip

\begin{definition}
Let
\[
J = \begin{pmatrix}
0 & 1 \\
-1 & 0
\end{pmatrix}.
\]
A meromorphic in $\bbC_+$ matrix function $\fA(z)$ is called $J$-\emph{contractive} if
$$
\frac{J-\fA(z)J\fA(z)^*}{z-\bar z}\ge 0.
$$
It is called $J$-\emph{inner} if in addition
$$
J-\fA(\l)J\fA(\l)^*=0\quad\text{for a.e.}\quad  \l\in\bbR.
$$

\end{definition}

The multiplicative theory of such matrix functions was founded by V.P. Potapov, see \cite{EP73}. For the case of \textit{entire} $2\times 2$-matrix functions, it was finalized by L. de Branges \cite{dB}.

In this subsection we provide a certain characterization of the \textit{defect subspace} $\cK_D=\hat \cH^2(\a)\ominus \cH^2_D$ by $J$-contractive matrix functions.
First we will show that $\cH^2_D$ possesses the following (de Branges spaces \cite{dB}) property.

\begin{lemma}\label{lem71}
If $f\in \cH^2_D$ and $f(z_0)=0$, $z_0\in\bbC_+$, then $g=\frac{z-\bar z_0}{z-z_0}f\in \cH^2_D$, moreover $\|f\|=\|g\|$.
\end{lemma}

\begin{proof}
Since $D$ is fixed, we simplify notation to $m_+$. Let $\s_+$ be the measure that provides its integral representation. By this definition we have
\begin{equation}\label{1-8feb21}
\frac{m_+(z)-\overline{m_+(w)}}{z-\bar w}=\int\frac{d\s_+(\l)}{(\l-z)(\l-\bar w)}=
\left\langle
\frac1 {\l-\bar w}, \frac1 {\l-\bar z}
\right\rangle_{L^2_{\s_+}}.
\end{equation}

Now we prove that
$$
(\cC F)(z) = L_D(z)\int\frac{F(\l)}{\l-z} \, d\s_+(\l)
$$
is a unitary map from $L^2_{\s_+}$ on $\cH^2_D$. As usual we start with reproducing kernels and note that
$$
\cC\left(\frac{1}{\l-\bar w}\overline{L_D(w)}\right)(z)=L_D(z)\frac{m_+(z)-\overline{m_+(w)}}{z-\bar w}\overline{L_D(w)}=k_D(z,w).
$$
By \eqref{1-8feb21} the map is norm preserving. It is well defined on a dense set and has all reproducing kernels in the image. Thus it is unitary.

If $f=\cC F$, then $f(z_0)=0$, $z_0\in\bbC_+$, is equivalent to $F$ being orthogonal to $1/(\l-\bar z_0)$.  For such $F$, we define
$$
G(\l) = \frac{\l-\bar z_0}{\l-z_0} F(\l).
$$
Note that $\|F\|_{L^2_{\s_+}} = \|G\|_{L^2_{\s_+}}$. At the same time for $z\not=z_0$ we get
\begin{align*}
(\cC G)(z) & = L_D(z) \int\frac{\l-\bar z_0}{\l-z_0}\frac{F(\l)}{\l-z} \, d\s_+(\l) \\
& = L_D(z)\int \left(\frac{z-\bar z_0}{z-z_0}\frac{1}{\l-z} + \frac{z_0-\bar z_0}{z_0-z}\frac{1}{\l-z_0}\right) F(\l) \, d\s_+(\l).
\end{align*}
Due to the orthogonality of $F$ to $1/(\l-\bar z_0)$, we have
$$
(\cC G)(z)=\frac{z-\bar z_0}{z-z_0}(\cC F)(z)=\frac{z-\bar z_0}{z-z_0}f(z)=g(z),
$$
and by analyticity the same holds for $z=z_0$. By its definition, $\cC G\in\cH^2_D$, and by unitarity, $\|\cC G\| = \|\cC F\|$.
\end{proof}

Recall that $\Phi_{z_0}(\infty)=1$ for an arbitrary $z_0\in \Omega$. Therefore
\begin{equation}\label{2-9feb21}
v=v_{z_0}(z)=\frac{z-z_0}{z-\bar z_0}=\frac{\Phi_{z_0}(z)}{\Phi_{\bar z_0}(z)}.
\end{equation}
In particular, this implies that the characters of both complex Green functions are the same and we denote this character by $\b_0=\b_{z_0}$.

By a \textit{unitary node} \cite{AGr83} we mean a unitary operator $U$ acting from $E_1\oplus K$ to $E_2 \oplus K$. The Hilbert space $K$ is called the \textit{state space} and the Euclidean spaces $E_1,E_2$ are called \textit{coefficient spaces}.

\begin{lemma}\label{lem72} Let $\hat k^\a_{z_0}$ be the reproducing kernel of $\hat \cH^2(\a)$.
Multiplication by $1/v$ in the decomposition
\begin{equation}\label{1-9feb21}
\frac 1{v}:
\left\{\frac{\hat k^{\a+\b_0}_{\bar z_0}}{\Phi_{\bar z_0}}\right\}\oplus \cK_D\oplus
\left\{{(k_D)_{z_0}}\right\}
\to
\left\{\frac{\hat k^{\a+\b_0}_{z_0}}{\Phi_{z_0}}\right\}\oplus \cK_D\oplus
\left\{(k_D)_{\bar z_0}\right\}
\end{equation}
forms a unitary node with the state space $\cK_D$.
\end{lemma}

\begin{proof}
According to \eqref{2-9feb21}, it is evident that the following  multiplication operator is unitary,
$$
\frac 1{v}:\frac 1{\Phi_{\bar z_0}}\hat\cH^2(\a+\b_0)\to \frac 1{\Phi_{ z_0}}\hat\cH^2(\a+\b_0).
$$
Since
$$
\frac 1{\Phi_{\bar z_0}}\hat\cH^2(\a+\b_0)=\left\{\frac{\hat k^{\a+\b_0}_{\bar z_0}}{\Phi_{\bar z_0}}\right\}\oplus\hat \cH^2(\a)
$$
and
$$
\frac 1{\Phi_{z_0}}\hat\cH^2(\a+\b_0)=\left\{\frac{\hat k^{\a+\b_0}_{ z_0}}{\Phi_{z_0}}\right\}\oplus \hat\cH^2(\a),
$$
we have a unitary node
$$
\frac 1{v}:\left\{\frac{\hat k^{\a+\b_0}_{\bar z_0}}{\Phi_{\bar z_0}}\right\}\oplus \hat\cH^2(\a)\to
\left\{\frac{\hat k^{\a+\b_0}_{ z_0}}{\Phi_{z_0}}\right\}\oplus \hat\cH^2(\a).
$$
By definition $\hat \cH^2(\a)=\cK_D\oplus\cH^2_D$. Let
$$
(\cH^2_D)_{z_0}=\{f\in \cH^2_D:\ f(z_0)=0\}.
$$
In these notations $1/v$ acts unitarily from
$$
\left\{\frac{\hat k^{\a+\b_0}_{\bar z_0}}{\Phi_{\bar z_0}}\right\}\oplus \cK_D\oplus \{(k_D)_{z_0}\}\oplus (\cH^2_D)_{z_0}
$$
to
$$
\left\{\frac{\hat k^{\a+\b_0}_{ z_0}}{\Phi_{ z_0}}\right\}\oplus \cK_D\oplus \{(k_D)_{\bar z_0}\}\oplus (\cH^2_D)_{\bar z_0}.
$$
By Lemma \ref{lem71}, the multiplication by $1/v$ acts unitarily in the last two components of these decompositions. Therefore it acts unitarily in their orthogonal complements and we have \eqref{1-9feb21}.
\end{proof}

We are now in a position to prove the main result of this section.

\begin{theorem}
Let $\hat D=\hat D(\a)$ be the divisor that corresponds to the space $\hat \cH^2(\a)$. Assume that
$D\not=\hat D$, but $\cA(D)=\a$. Then the following relation uniquely defines a $J$-contractive matrix function $\fA_D(z)$,
\begin{equation}\label{1-10feb21}
\cL_{\hat D}\fA_D(z)=\cL_D(z),
\end{equation}
which is entire with respect to $1/z$.
\end{theorem}

\begin{proof}
Theorem 4.7 from \cite{BLY2} can be word-by-word adapted to the current situation (mainly according to Lemma \ref{lem72}). Since $\Omega$ does not obey DCT, we cannot use \cite[Lemma 4.8]{BLY2}. But the set $\sE$ is a system of intervals that accumulates to the origin only. For this reason, classical complex analysis methods can be used to show that $\fA_D(z)$ is analytic on the open spectral intervals (see the symmetry property (c) \cite[Theorem 4.7]{BLY2}), and then that all possible isolated singularities at the end points $a_j,b_j$, $j\ge 1$, are removable. Since the matrix is also analytic at infinity, the only possible singularity is the origin. Thus $\fA_{D}(z)$ is entire w.r.t. $1/z$.
\end{proof}

\subsection{The Hard Part}\label{sect63}

For $\sE$ with a minimal violation of DCT, we get that  $\fA_D(z)$ is a polynomial with respect to $1/z$.

\begin{theorem}[Potapov \cite{EP73}]  Assume that $\fA$ is a $J$-inner (contractive in the upper half-plane) matrix-function that is a polynomial w.r.t. $1/z$. Then $\fA(z)$ admits the following product representation,
\begin{equation}\label{1-19feb21}
\fA(z)=\fA(\infty)\prod_{j=1}^n\fA_j(z),\quad \fA_j(z)=\left(I-\frac 1 z\cE_jJ\right ),
\end{equation}
where $\cE_j\ge 0$ and $(\cE_j J)^2=0$. In other words, the matrices $\cE_j$ are of the form
$$
\cE_j=\rho_j\begin{pmatrix}
\cos\varphi_j\\ \sin\varphi_j
\end{pmatrix}
\begin{pmatrix}
\cos\varphi_j&\sin\varphi_j
\end{pmatrix},\quad \phi_j\in\bbR,\ \rho_j>0.
$$
\end{theorem}

Note that
$$
\frac{J-\fA_j(z)J\fA_j(z)^*}{z-\bar z}=\frac 1{|z|^2}\cE_j\ge 0,
$$
that is, every such product represents a $J$-inner polynomial.
We choose the shortest such factorization, i.e., the consecutive $\varphi_j$ are distinct: otherwise two consecutive terms can be combined into one factor.

\begin{corollary}\label{c11oct21}
Assume that $\fA_D(z)= \mathrm{const}$. Then, $\hat D=D$.
\end{corollary}

\begin{proof}
This constant is necessarily of the form
$$
\fA_D(z)=\begin{pmatrix}
1&0\\ h_D&1
\end{pmatrix}.
$$
In the bottom row we have
$$
\lim_{z\to 0}\begin{pmatrix}
(n_{\hat D})_+(z)&1
\end{pmatrix}\begin{pmatrix}
1&0\\ h_D&1
\end{pmatrix}=
\lim_{z\to 0}\frac{L_D(z)}{L_{\hat D}(z)}
\begin{pmatrix}
(n_{D})_+(z)&1
\end{pmatrix}.
$$
Therefore,
$$
\lim_{z\to 0}\frac{L_D(z)}{L_{\hat D}(z)}=1.
$$
By definition, $(n_D)_+(0)=(n_{\hat D})_+(0)=0$. Thus $h_D=0$, $L_{\hat D}= L_D$, and consequently $\hat D= D$.
\end{proof}

\begin{lemma} \label{lemma2mar1}
$D \in \cD_b(\E)$ if and only if $n_\pm$ have asymptotic behavior
\begin{align}
n_+(z) & =  \sigma_+ z + o(\lvert z \rvert), \qquad z \uparrow 0 \label{24feb3} \\
n_-(z) & =  \sigma_- z + o(\lvert z \rvert), \qquad z \uparrow 0  \label{24feb4}
\end{align}
for some $\sigma_+, \sigma_- > 0$.
\end{lemma}

\begin{proof}
$D \in \cD_b(\E)$ if and only if $R$ has a pure point at $0$. Denote its weight by $\sigma_0 > 0$.  The point mass can be recovered as
\[
\sigma_0 = \lim_{z\uparrow 0} (-z) R(z),
\]
so $D \in \cD_b(\E)$ is equivalent to the claim that $R$ has asymptotic behavior
\begin{equation}\label{24feb5}
- R(z)^{-1} = \sigma_0^{-1} z + o(\lvert z \rvert), \qquad z \uparrow 0
\end{equation}
for some finite $\sigma_0 > 0$.

Adding together \eqref{24feb3}, \eqref{24feb4} gives \eqref{24feb5}. Conversely, assume $D \in \cD_b(\E)$. Then $-R(z)^{-1} \to 0$ as $z \uparrow 0$ implies $n_-(0) = 0$.

By the Herglotz representation for $n_+$,
\[
n_+(z) = n_+(-1) + \int_{[0,\infty)} \left( \frac{1}{\lambda -z} - \frac{1}{\lambda+1} \right) d \nu_+(\lambda).
\]
By monotone convergence,
\[
0 = \lim_{z\uparrow 0} n_+(z) = n_+(-1) + \int_{[0,\infty)} \left( \frac{1}{\lambda} - \frac{1}{\lambda+1} \right) d\nu_+(\lambda).
\]
Dividing by $z$ gives
\[
\frac{n_+(z)}{z} = \int_{[0,\infty)} \frac{1}{ \lambda (\lambda - z)} \, d\nu_+(\lambda),
\]
so by another monotone convergence,
\[
\sigma_+ := \lim_{z \uparrow 0} \frac{n_+(z) }{z} = \int_{[0,\infty)} \frac{1}{ \lambda^2} \, d\nu_+(\lambda).
\]
This is strictly positive, because $n_+$ is not a constant function. Moreover, $\sigma_+$ cannot be infinite, because it is majorized by the directional derivative of $-R^{-1}$: due to monotonicity of $n_-$ on $(-\infty,0)$,
\[
\frac{n_+(z) }{z} \le \frac{- R(z)^{-1} + R(0)^{-1}}{z}
\]
and the latter has a finite limit by \eqref{24feb5}. This proves \eqref{24feb3}; the proof of \eqref{24feb4} is analogous.
\end{proof}

\begin{lemma}
If for some divisor $D \in \cD(\E)$, at least one of the functions $n_\pm$ has asymptotics given by \eqref{24feb3} or \eqref{24feb4}, then there exists a divisor $\tilde D \in \cD_b(\E)$, $\tilde D \neq D$, such that $\cA(\tilde D) = \cA(D)$. In particular, then $\cA(D) \in \cA(\cD_b(\E))\subset \Xi$.
\end{lemma}

\begin{proof}
Without loss of generality, we assume \eqref{24feb3}. Denote
\[
\fA = I - \frac 1z \cE J 
\]
where
\[
\cE = \rho \begin{pmatrix}
1 \\ 0
\end{pmatrix}
\begin{pmatrix}
1 & 0
\end{pmatrix}, 
\]
$\rho > 0$. Define $r_\pm$ by
\begin{align*}
\begin{pmatrix}
r_+ & 1
\end{pmatrix} &  \simeq \begin{pmatrix}
n_+ & 1
\end{pmatrix} \fA, \\
\begin{pmatrix}
-r_- & 1
\end{pmatrix} &  \simeq \begin{pmatrix}
-n_- & 1
\end{pmatrix} \fA.
\end{align*}

To compute the asymptotics at $0$, we write explicitly
\[
r_+(z) = \frac{ n_+(z)  }{1 -  \frac \rho z n_+(z)  },
\]
so
\begin{equation}\label{28feb1}
- \frac {1}{ r_+(z) } = -  \frac {1}{n_+(z) } +    \frac{ \rho}z.
\end{equation}
Note $-\frac {1}{n_+(z) }$ is a Herglotz function with a point mass at $0$ of mass $\sigma_+^{-1}$. Thus, if
\[
0 < \rho <  \sigma_+^{-1},
\]
then \eqref{28feb1} defines another Herglotz function with a point mass at $0$ of mass $\sigma_+^{-1} - \rho > 0$. Therefore $r_+$ is a Herglotz function with the asymptotics at $0$ given by
\[
r_+(z) =  \tilde \sigma_+ z + o(\lvert z\rvert), \qquad z \uparrow 0,  \qquad \tilde \sigma_+ = \frac 1{\sigma_+^{-1} - \rho}.
\]
The function $r_-$ is a Herglotz function because $\fA$ is $J$-contractive. Its asymptotics at $0$ is computed similarly,
\begin{equation}\label{28feb2}
- \frac {1}{ r_-(z) } = -  \frac {1}{n_-(z) } - \frac{ \rho }z.
\end{equation}
The right-hand side has a pole at $0$ of weight at least $\rho$. Denoting the reciprocal of that weight by $\tilde\sigma_-$ leads to
\[
r_-(z) = \tilde\sigma_- z + o(\lvert z\rvert), \qquad z \uparrow 0.
\]
The reflectionless property for the pair $r_+, r_-$ follows from that of $n_+, n_-$ by definition. Since they are distinct, they correspond to a different divisor $\tilde D \neq D$.

Since $\fA$ is single-valued, these divisors correspond to the same character.
\end{proof}

\begin{proof}[Finishing the Proof of Theorem \ref{l13sept21}]
Assume $\alpha \in \Xi \setminus \cA(\cD_b(\E))$. Then there exists $D$ such that $\cA(D) = \alpha$ and $D \neq \hat D$. The matrix function  $\fA_D$ is well defined and it is non-constant. Thus, from
\begin{align*}
\begin{pmatrix}
 \hat n_+ & 1
\end{pmatrix} &  \simeq \begin{pmatrix}
 n_+ & 1
\end{pmatrix} \fA_D^{-1} \\
\begin{pmatrix}
- \hat n_- & 1
\end{pmatrix} &  \simeq \begin{pmatrix}
- n_- & 1
\end{pmatrix} \fA_D^{-1}
\end{align*}
we should get for some intermediate spaces the two-term asymptotics
\begin{align*}
r_+(z) & = w_+ + \sigma_+ z + o(\lvert z \rvert) \qquad z \uparrow 0 \\
r_-(z) & = w_- + \sigma_- z + o(\lvert z \rvert) \qquad z \uparrow 0.
\end{align*}
To avoid working with $\fA(\infty)$, it is more convenient to compute backwards:
\begin{align*}
\begin{pmatrix}
r_+ & 1
\end{pmatrix} &  \simeq \begin{pmatrix}
 n_+ & 1
\end{pmatrix}  \fA_k^{-1} \\
\begin{pmatrix}
- r_- & 1
\end{pmatrix} &  \simeq \begin{pmatrix}
- n_- & 1
\end{pmatrix} \fA_k^{-1},
\end{align*}
where $k$ denotes the degree of $\fA_D$ as a polynomial in $1/z$, see \eqref{1-19feb21}.
We write explicitly
\[
r_+(z) = \frac{z n_+(z)  - \rho_k \sin \varphi_k \cos \varphi_k  n_+(z) -  \rho_k \sin^2 \varphi_k}{ \rho_k\cos^2 \varphi_k n_+(z) + z +  \rho_k \sin \varphi_k \cos \varphi_k},
\]
and therefore, if $\varphi_k \neq 0, \pi/2$,
\begin{align*}
r_+(z) & = - \tan \varphi_k  \left[ 1 - \frac{ \frac{zn_+(z)}{\sin^2 \varphi_k} + \frac{z}{\sin\varphi_k \cos\varphi_k}}{  \rho_k + \rho_k \cot \varphi_k n_+(z) + \frac z{\sin\varphi_k \cos\varphi_k} } \right]
\\
& = - \tan \varphi_k + \frac{z}{ \rho_k \cos^2 \varphi_k} + o(z).
\end{align*}

Similarly, if $\varphi_k = 0$, we compute
\[
r_+(z) = \frac{z n_+(z)}{\rho_k n_+(z) + z} = \frac{z}{ \rho_k + z/n_+(z)} = \frac{z}{ \rho_k} + o(z).
\]
Thus, in both cases, $r_+$ is of the form
\[
r_+(z) = - \tan \varphi_k +  \frac{z}{ \rho_k \cos^2 \varphi_k}  + o(z), \qquad z \uparrow 0.
\]
If $\varphi_k = \pi/2$, we similarly compute
\[
r_+(z) = \frac{z n_+(z)- \rho_k}{z} = - \frac { \rho_k}z + o(1).
\]
In this case, by comparing with the behavior of $\hat n_+$ which has a finite limit at $0$, and since $\fA(\infty)$ does not affect the finiteness of that limit, we can conclude by contradiction that $k\neq 1$. Thus, if $\varphi_k = \pi/2$, we can apply a second matrix,
\begin{align*}
\begin{pmatrix}
\tilde r_+ & 1
\end{pmatrix} &  \simeq \begin{pmatrix}
 n_+ & 1
\end{pmatrix}  \fA_k^{-1} \fA_{k-1}^{-1} \\
\begin{pmatrix}
- \tilde r_- & 1
\end{pmatrix} &  \simeq \begin{pmatrix}
- n_- & 1
\end{pmatrix} \fA_k^{-1} \fA_{k-1}^{-1}.
\end{align*}
Then similar calculations from $r_+(z) = -1/z + o(1)$ give
\[
\tilde r_+(z) = \frac{z r_+(z)  -  \rho_{k-1} \sin \varphi_{k-1} \cos \varphi_{k-1}  r_+(z) - \rho_{k-1} \sin^2 \varphi_{k-1}}{\rho_{k-1} \cos^2 \varphi_{k-1} r_+(z) + z + \rho_{k-1}  \sin \varphi_{k-1} \cos \varphi_{k-1} },
\]
and then
\[
\tilde r_+(z) = - \tan \varphi_{k-1} + \frac{z}{\rho_{k-1} \cos^2 \varphi_{k-1}} + o(z). 
\]
Analogous calculations will give expansions for $r_-$ and $\tilde r_-$.

\bigskip

Now that we have those asymptotics for $r_\pm$:  with our conventions we have $w_+ = 0$ and, since $\alpha \notin \cD_b(\E)$, Lemma~\ref{lemma2mar1} implies $w_- \neq 0$. Then
\[
- R(z)^{-1} = w + \sigma z
\]
with $w = w_- < 0$ and $\sigma = \sigma_+ + \sigma_-$.

The map $-R^{-1} - w$ is a Herglotz function and so is the map $\tau \mapsto \sqrt{ - R^{-1}(\tau^2) - w}$. As an elementary consequence of the Herglotz representation, a nontrivial Herglotz function decays at most linearly in a normal boundary limit, so
\[
\lim_{z \uparrow 0} (- R(z)^{-1} )' = \sigma \neq 0,
\]
which implies
\[
\lim_{z \uparrow 0}  \left( \log (- R(z)^{-1} ) \right)' = \frac\sigma w \neq 0.
\]
By writing the representation of $ \log (- R(z)^{-1} )$, differentiating, and argument considerations, we conclude that
\[
\int_{\E} \frac {d\lambda}{\lambda^2}  < \infty,
\]
which is a contradiction, see \eqref{43sept21}. Thus,  $\Xi \setminus \cA(\cD_b(\E)) = \emptyset$, which concludes the proof.
\end{proof}

\section{The Basic Spectral Set $\cE$. Proposition \ref{pr17sept21} and Beyond}\label{sect7}

We construct  a spectral set $\cE$ of Widom type with a minimal violation of DCT  using ideas of the previous section, see especially Sect.~\ref{sect63}. We multiply the Blaschke-Potapov factor
$$
\fA_\rho(\mu)=I+2\rho\mu
\begin{pmatrix}
1&0  \\
0&0
\end{pmatrix}J,\quad \rho>0
$$
by the simplest $J$-contractive entire matrix function
$$
\fB(\mu)=\begin{pmatrix}
\cos\pi\mu&\sin\pi\mu  \\
-\sin\pi\mu&\cos\pi\mu
\end{pmatrix}.
$$
The following lemma is a consequence of the general theory of entire $J$-contractive matrix functions, but it can also be easily checked by direct computations. For notations see Sect.~\ref{sbs32}.

\begin{lemma} The trace of the product $\fA_\rho\fB$ is of the form
\begin{equation}\label{14oct21}
\frac 1 2 \tr \fA_\rho(\mu)\fB(\mu)=\cos\pi\Delta_\rho(\mu),
\end{equation}
where $\Delta_\rho(\mu)$ is a conformal mapping of $\bbC_+$ on a comb domain $\Pi_0=\Pi_0(\rho)$ with the frequencies $n_k=k$, $k\in\bbZ\setminus \{0\}$.
\end{lemma}
Using direct computations we can characterize both  spectral-  and comb-parameters of
$\Delta_\rho(\mu)$ by their explicit asymptotics.
\begin{lemma}\label{l7221}
For the comb-function $\Delta_\rho$ defined by \eqref{14oct21} we have \eqref{130jun21}. The gap parameters $\mu_k^{-},\mu_k^{+}$, $k\in\bbZ_+$,  are given by
\begin{equation}\label{24oct21}
\mu^+_{k}=k,\quad \mu_{k  +1}^-\simeq k+\frac{2}{\pi\rho k},
\end{equation}
and the heights of slits by $\u_k\simeq\log (2k+1)\rho$.
\end{lemma}

\begin{proof}[Proof of Proposition \ref{pr17sept21}] Recall  that
$$
\cE=\left\{z=\frac 1{\mu^2}:\ |\cos\pi\Delta_\rho(\mu)|\le 1\right\}.
$$
In particular $a_k=1/(\mu_k^+)^2$, $b_k=1/(\mu_k^-)^2$, $k\ge 1$. By \eqref{24oct21},
$$
\log\frac{a_k}{b_{k+1}}\simeq 2\log \left(1+\frac{2}{\pi\rho k^2}\right)\simeq
\frac{4}{\pi\rho k^2},
$$
that is, we have \eqref{33sept21}. Similarly, by
$$
\frac 1{b_{k+1}}-\frac{1}{a_k}\simeq \left(k+\frac{2}{\pi\rho k}\right)^2-k^2\simeq\frac{4}{\pi\rho}
$$
we get \eqref{43sept21}.

 In our construction the domain
$$
\Omega_\mu=(\bbC\setminus\bbR)\cup_{k\not=0}(\mu_k^-,\mu_k^+)
$$
is essentially defined as the domain in which
$\Im\Delta_\rho(\mu)$ is  the Martin function w.r.t.  infinity.
Moreover, it follows immediately from its explicit representation \eqref{130jun21} that it behaves as $\Im\mu$ at infinity. By Koosis' terminology this is the Phragm\'en-Lindel\"of function for the domain. According to Theorem
\cite[p. 407]{Koos}, for the Green function
$G(\mu, i,\Omega_\mu)$ of  $\Omega_\mu$ w.r.t. $i$ we have
$$
\int_\bbR G(\mu, i,\Omega_\mu)d\mu<\infty.
$$
Since the lengths of gaps are uniformly bounded from below we have the Widom condition
$$
\sum_{\nu: \nabla G(\mu, i,\Omega_\mu)=0}G(\nu, i,\Omega_\mu)<\infty.
$$
Respectively, $\bbC\setminus\cE_\rho$ is also of Widom type.
\end{proof}

\begin{proof}[Speculation on a Proof of Conjecture \ref{conj33}]
For the set $\cE_\rho$  and the complex Martin function
$(\Theta_\rho)_0$ defined in  \eqref{76sept21}, consider $-1/(\Theta_\rho)_0(z)$ in the upper half plane $\bbC_+$. The function maps conformally $\bbC_+$ on $\bbC_{++}$ with slits along the circles orthogonal to the real axis (\textit{curved} slits).

We point out that the bases of the slits are  $1/k$ and their ``heights", see Lemma \ref{l7221},
are given by
$$
\max_{\l\in(a_k,b_k)}{\Im}\frac{-1}{(\Theta_\rho)_0(\l)}=\frac{\u_k}{k^2+\u_k^2}\simeq\frac{\log(2k+1)\rho}{k^2}.
$$
These slits are asymptotically straight. For the deviation  of the top- to  the base-point in the real direction we have
$$
\frac 1{k}-\min_{\l\in(a_k,b_k)}{\Re}\frac{-1}{(\Theta_\rho)_0(\l)}
=\frac{\u_k}{k}\frac{\u_k}{k^2+\u_k^2},
$$
which is $\frac{\log k}k$-times smaller than the corresponding ``height".

Now,
let  $\hat\Pi$ be the domain with  the \textit{same bases} $\hat\eta_k=1/k$ but \textit{straight} slits.
Let  $\hat\Theta$ be the conformal mapping $\bbC_+\to\hat \Pi$ (infinity goes to infinity),
$\hat \sE=\hat\Theta^{-1}(\bbR_+)$.
We conjecture that with suitable heights
$$
\hat h_k\sim\frac{\log k}{k^2 },
$$
there exists a quasi-conformal map $\cQ$ that maps our domain with curved slits to $\hat \Pi$, i.e.,
$$
\cQ:\ \{w:\ -1/w\in \Pi_0^-\}\to \hat \Pi,
$$
so that the composition  $\hat\Theta^{-1}\circ\cQ\circ(-1/\Theta_0)$ maps $\cE\to\hat\sE$ and the resulting domain $\hat\Omega=\bbC_+\setminus\hat \sE$ inherits  the properties of the domain $\Omega_{\cE}$. That is,
$\hat\Omega$ is of Widom type with a minimal DCT violation and $\hat\sE$ obeys \eqref{33sept21}, \eqref{43sept21}, as it was stated in Conjecture \ref{conj33}.
\end{proof}

\begin{remark}
The best known criterion for DCT is the homogeneity property: if there exists $C>0$ such that
$$
|\sE\cap(x-\delta, x+\delta)|\ge C\delta
$$
for an arbitrary $x\in\sE$ and $\delta>0$, then $\Omega$ is of Widom type and DCT holds.
Thus, if we have a homogeneous set $\sE$ and bi-Lipschitz $\cQ:\bbR\to\bbR$, then
$\hat\sE=\cQ(\sE)$ is also homogeneous. That is, Widom and DCT properties are inherited for such transforms. This provides some hope that one can find a suitable class of qc-mappings (special properties of their restrictions on the real axis) so that the above speculation becomes rigorous.
\end{remark}

\begin{proof}[Proof of Lemma \ref{l16sept21}]
Let
$$
\cT^0_\eta=\{\a\in\cT_\eta:\ \a_1=0\}\simeq \hat \bbZ.
$$
The algebra of this adic-part of $\cT_\eta$ in this case deals with identities$$
n_1\a_{n_1 n_2}=\a_{n_2}, \quad \a_1=0.
$$
For this reason, for an arbitrary prime $p$, we get that the sequence $\a_{p^s}$ is generated by a $p$-adic number  (see for details the proof of Proposition~\ref{pr16sept21})
$$
\delta_p=\{\delta_p^{(1)}, \delta_p^{(2)}, \delta_p^{(3)},\dots \},\quad
\delta_p^{(k)}\in\bbZ \mod p.
$$
In turn, for an arbitrary $k=p_1^{s_1}\cdots p_n^{s_n}$, we have
$$
\frac{k}{p_j^{s_j}}\a_k=\a_{p_j^{s_j}} \mod 1\quad j=1,\dots, n.
$$
The last collection of identities defines $\a_k$ uniquely (Chinese remainder theorem). Thus $\a\in \cT^0_\eta$ is defined uniquely by the sequence
$\{\delta_p\}\in\hat \bbZ$. Taking into account the first component $\a_1\in\bbR/\bbZ$, we get
$\cT_\eta\simeq(\bbR/\bbZ)\times\hat \bbZ$.
\end{proof}

\section{The Main Lemma \ref{mainl} and the Main Obstacle}

\begin{proof}[Proof of Lemma \ref{mainl}]
As $\cT_\eta(\a_V )\cap \Xi = \emptyset$, in particular $\cB(V)\not\in \cD_b(\sE)$, we can apply
Proposition \ref{prop13sept21}, according to which
$$
(\cA\circ\cB)(V(x+x_0))=\a_V-\eta x_0\in \cT_\eta(\a_V).
$$
By definition, $\Xi$ is  the set of characters with non-unique preimages under $\mathcal{A}$, the restriction of $\mathcal{A}$ to $\mathcal{A}^{-1}(\cT_\eta(\a_V))$ is a continuous bijection of compact sets and hence its inverse
$$
\mathcal{A}_{\cT_\eta(\a_V)}^{-1} : \cT_\eta(\a_V) \to \mathcal{A}^{-1}(\cT_\eta(\a_V)), \; \alpha \mapsto \mathcal{A}^{-1}(\alpha)
$$
is continuous.

By asymptotic relations \eqref{Rinfinity2} and \eqref{e.RDasymptotics}
$$
V_D(0)=\frac 1 2 \cQ_1(D)
$$
(a well known trace formula). We point out that $\cQ_1(D)$ is continuous on $\cD(\sE)$.
Therefore  $\frac 1 2 \cQ_1 \circ \mathcal{A}_{ \cT_\eta(\a_V)}^{-1}$ is continuous on $\cT_\eta(\a_V)$, that is, up to the shift, it can  be treated as a
continuous function
 given on  the compact abelian group $\cT_\eta$.

For every $x \in \mathbb{R}$, $V(x)$ is the image of the element $\alpha_V - \eta x$.
 This shows that $V$ is almost periodic with the sampling function $\frac 1 2 \cQ_1 \circ \mathcal{A}_{ \cT_\eta(\a_V)}^{-1}$.
\end{proof}

\begin{remark} Since the manuscript contains both conjectures and proved propositions, we would like to point out that
the above proof of our main Lemma \ref{mainl} is based on Proposition \ref{prop13sept21} (that was proved), while this proposition itself is only a part of the general Conjecture \ref{th13sept21} (which was not proved here in full generality).
\end{remark}

\begin{remark}
The lemma is simple but seems highly important independently of the Deift conjecture. It provides a reason to have almost periodic potentials with resolvent domains that violate DCT. Let us mention a parallel activity in a description of asymptotic behaviour of Chebyshev polynomials. In \cite[Theorem 1.5]{CSYZ} it was shown that
an almost periodic behavior of the second term in the asymptotics for the minimal deviation of Chebyshev polynomials implies DCT for generic (regular in the terminology of this paper) sets. The authors raised a question: what happens for non-generic sets? The same argument of non-intersection of
a closure of shifts at the given direction with the set of singular characters would provide examples of almost periodicity with violation of DCT in Chebyshev's problem.
\end{remark}

Lemma~\ref{mainl} has a conditional character: as  the main assumption it requires that the closure of shifts $\cT_\eta(\a_V)$ and the set of singular characters  $\Xi$ do not intersect. Our description of the set $\cT_\eta(\a_V)$ (Conjecture \ref{conj33}, Lemma \ref{l16sept21}) most likely is optimal. However, we have a very limited understanding of how to transform our knowledge of the set $\Xi$ given in terms of divisors (Theorem \ref{l13sept21}) in terms of characters, so that we would be able to check the intersection of these two sets. We consider this as \textbf{the main obstacle} in implementing our program to disprove the Deift conjecture (up to this, proving the other conjectures stated here is a matter of reasonable time and mental efforts).

Essentially simplifying the problem, we introduced the set $\Xi^0$  with \textit{an explicit characterization in terms of characters} \eqref{16oct21}. Note that in Proposition~\ref{pr16sept21} we also essentially simplified the torus $\cT_\eta$ in comparison with the one we really expect according to Conjecture \ref{conj33}. The reason for the last simplification was just to avoid an involved object in the introductory Section \ref{sect2}. We prove now Proposition \ref{pr16sept21} as it was stated. In Corollary \ref{c16oct}, with minimal effort, we will get the same result in the case $\eta_k=k$.

\begin{proof}[Proof of Proposition \ref{pr16sept21}]
We have $x_j\eta_1\to \a_1\mod 1$. We fix $\a_1\in [0,1)$. Furthermore,
for the components of $\cT_\eta$ before the closure, we have
$$
3x_j \eta_{k+1}=x_j \eta_k \mod 1.
$$
Respectively, in the limit we get the same relation,
$$
3\a_{k+1}=\a_k \mod 1.
$$
Therefore,
$$
\a_{k+1}=\frac{\e_k}{3}+\frac{\a_k} 3=\frac{\e_k}{3}+\frac{\e_{k-1}}{3^2}+\dots+\frac{\e_1}{3^{k}}+\frac{\a_1}{3^{k}},
\quad \e_k\in[0,1,2].
$$
Thus, as the parameters for $\a\in\cT_\eta$ we get $\a_1\in\bbR/\bbZ$ and the sequence
$$
\e_1+3\e_2+\dots+\e_k 3^{k-1},
$$
which forms a triadic integer.

Note that $\a_1/3^k\to 0$. So, we need to check the intersection of $\Xi^0$ with $\cT_\eta^0(\b)=\{\a+\b:\ \a\in\cT_\eta, \a_1=0\}$.
If the intersection is not empty, using
$$
\b_{k}=\frac{\b^{(k)}_{1}} 3+\frac{\b^{(k)}_{2}} {3^2}+\dots,
$$
we have
$$
\b_{k+1}+\a_{k+1}=\frac{\b^{(k+1)}_{1}+\e_k} 3+\frac{\b^{(k+1)}_{2}+\e_{k-1}} {3^2}+\dots\to 0.
$$
Therefore, for sufficiently big $k$,
$$
\b_1^{(k+1)}=-\e_k,\ \b_2^{(k+1)}=-\e_{k-1}=\b^{(k)}_1.
$$
As soon as $\b_2^{(k+1)}\not=\b^{(k)}_1$ for all $k$, we get a contradiction.
\end{proof}

\begin{corollary}\label{c16oct} Let $\cT_\eta$ be generated by
 the system of frequencies  $\eta_k = 1/k$.
There exists $\beta \in \pi_1(\Omega)^*$ such that $\cT_\eta(\b) \cap \Xi^0 = \emptyset$.
\end{corollary}

\begin{proof}
Considering the components $\alpha_n$ with $n = 3^k$, we see that
\[
\alpha_{3^{k+1}} = \frac{\e_k}{3}+\frac{\e_{k-1}}{3^2}+\dots+\frac{\e_1}{3^{k}}+\frac{\a_1}{3^{k}},
\quad \e_k\in[0,1,2],
\]
so those components by themselves form a $3$-adic integer.

Moreover, if $\alpha_n + \beta_n \to 0$ as $n \to \infty$, then
\[
\alpha_{3^k} + \beta_{3^k} \to 0, \quad k \to \infty
\]
so the proposition follows from the previous model problem.
\end{proof}

\end{document}